\documentclass[11pt,letterpaper]{article}

\usepackage{typearea}
\paperwidth 8.5in \paperheight 11in \typearea{15}

\usepackage{theorem,latexsym,graphicx, floatflt}
\usepackage{amsmath,amssymb,enumerate,amsfonts}
\usepackage{xspace}
\usepackage{bm}
\usepackage{ifpdf}
\usepackage{shadow,shadethm,color}
\usepackage{algorithm}
\usepackage[noend]{algorithmic}
\usepackage{subfigure}
\usepackage{verbatim}
\usepackage{paralist}
\usepackage{algorithm,algorithmic}

\allowdisplaybreaks

\definecolor{Darkblue}{rgb}{0,0,0.4}
\definecolor{Brown}{cmyk}{0,0.81,1.,0.60}
\definecolor{Purple}{cmyk}{0.45,0.86,0,0}
\newcommand{\mydriver}{hypertex}
\ifpdf
 \renewcommand{\mydriver}{pdftex}
\fi
\usepackage[breaklinks,\mydriver]{hyperref}
\hypersetup{colorlinks=true,%
            citebordercolor={.6 .6 .6},linkbordercolor={.6 .6 .6},%
citecolor=Darkblue,urlcolor=black,linkcolor=Darkblue,pagecolor=black}
\newcommand{\lref}[2][]{\hyperref[#2]{#1~\ref*{#2}}}

\newtheorem{theorem}{Theorem}[section]

\newtheorem{conjecture}[theorem]{Conjecture}
\newtheorem{lemma}[theorem]{Lemma}
\newshadetheorem{lemmashaded}[theorem]{Lemma}

\newtheorem{cl}[theorem]{Claim}

\numberwithin{algorithm}{section}

\newenvironment{proof}{

\noindent{\bf Proof:}} {\hfill$\blacksquare$

}

\newcommand{\junk}[1]{}
\newcommand{\ignore}[1]{}

\newcommand{\D}[0]{{\ensuremath{\mathcal{D}}}\xspace}
\newcommand{\I}[0]{{\ensuremath{\mathcal{I}}}\xspace}
\newcommand{\J}[0]{{\ensuremath{\mathcal{J}}}\xspace}
\newcommand{\oq}[0]{{\ensuremath{\overline{q}}}\xspace}

\newcommand{\Ex}{\mathbb{E}}

\newcommand{\OPT}{\ensuremath{{\sf opt}}\xspace}

\newcommand{\svrp}{\ensuremath{{\sf StochVRP}}\xspace}
\newcommand{\obj}{\ensuremath{{\sf obj}}\xspace}
\newcommand{\sobj}{\ensuremath{\widehat{\sf obj}}\xspace}
\newcommand{\aug}{\ensuremath{{\sf Aug}}\xspace}
\newcommand{\pair}[2]{{\ensuremath{\langle #1, #2\rangle}}\xspace}
\newcommand{\mst}{\ensuremath{{\sf MST}}\xspace}
\newcommand{\flow}{\ensuremath{{\sf Flow}}\xspace}
\newcommand{\lp}{\ensuremath{{\sf LP}}\xspace}
\newcommand{\kro}{\ensuremath{{\sf KnapRankOrient}}\xspace}
\newcommand{\sop}{\ensuremath{{\sf SOP}}\xspace}
\newcommand{\llog}{\ensuremath{\log\!\log}}

\newcommand{\sse}{\subseteq}

\newcommand{\initOneLiners}{%
    \setlength{\itemsep}{0pt}
    \setlength{\parsep }{0pt}
    \setlength{\topsep }{0pt}
}
\newenvironment{OneLiners}[1][\ensuremath{\bullet}]
    {\begin{list}
        {#1}
        {\initOneLiners}}
    {\end{list}}

\begin{document}

\title{Stochastic Vehicle Routing with Recourse}
\author{Inge Li G{\o}rtz\thanks{Technical University of Denmark, DTU Informatics.}
\and Viswanath Nagarajan \thanks{IBM T.J. Watson Research Center.}  \and Rishi Saket\footnotemark[2]
 }

\date{}
\maketitle
\begin{abstract}
We study the  classic {\em Vehicle Routing Problem} in the setting of stochastic optimization with recourse. \svrp is a
two-stage optimization problem, where demand is satisfied using two routes: fixed and recourse. The fixed route is
computed using only a demand distribution. Then after observing the demand instantiations, a recourse route is computed
-- but costs here become more expensive by a factor $\lambda$.

We present an $O(\log^{2}n \cdot \log(n\lambda))$-approximation algorithm for this stochastic routing problem, under
arbitrary distributions. The main idea in this result is relating \svrp to
 a special case of {\em submodular
orienteering}, called {\em knapsack rank-function orienteering}. We also give a better approximation ratio for {\em knapsack rank-function orienteering} than what follows from prior work. Finally, we provide a Unique Games Conjecture based $\omega(1)$ hardness of
approximation for \svrp, even on star-like metrics on which our algorithm achieves a logarithmic approximation.
\end{abstract}

\section{Introduction} Consider a distribution problem involving a depot location and a set
of customer locations. There is a vehicle of capacity $Q$ that is used to distribute items. The demand at customer
locations is random with a known (joint) distribution \D. The distributor wants to plan a {\em fixed route} for this
capacitated vehicle, that will be employed on a daily basis. However due to the stochastic nature of demands, the fixed
route might be insufficient to meet all demands. Therefore the distributor also plans a secondary {\em recourse
strategy}, that satisfies all unmet demands after the fixed route. Each morning the distributor receives the precise
demand quantities from all customers (drawn from \D). Based on this he/she decides which subset of customers will be
satisfied along the fixed route, and then plans a recourse route to satisfy the remaining customers. The goal is to
minimize the cost of the fixed route plus the expected cost of the recourse route. Examples of real-world applications
are local deposit collection from bank branches, garbage collection, home heating oil delivery, and forklift
routing~\cite{AE07,B92}.

A solution based on fixed routes is desirable for several reasons, and is commonly used in practice;
see~\cite{SG95,ESU09} for more detailed discussions on this. In our context, there are at least two advantages. First,
the driver can get familiar with the road/traffic conditions which results in time savings. Moreover, having fixed
routes simplifies the everyday route planning process: the incremental recourse step will typically contain fewer
demands.

Fixed-route problems are often modeled in the framework of two-stage stochastic optimization. {\em A priori}
optimization handles some natural but simple recourse strategies: eg., short-cutting over customers without demand in
TSP~\cite{BJO90,ST08}, and refill-visits from the depot in the Vehicle Routing Problem (VRP)~\cite{B92,GNR12}.
Recently, more complex recourse actions have been considered: adding penalty terms in deadline TSP~\cite{CT08}, and
using backup vehicles in VRP~\cite{AE07}.

In this paper, we penalize the cost of the recourse route by an inflation factor $\lambda\ge 1$. This is also a common
approach for two-stage stochastic optimization with recourse. Furthermore, in the stochastic VRP we consider, recourse
strategies are non-trivial since it also involves choosing the subset of realized demands served by the fixed route. In
this respect it is unlike most previously studied 2-stage stochastic problems (eg.~\cite{RS06,SS06,GPRS11}) where the
recourse step is just a deterministic instance of the same problem. Before describing the results of this paper, we
define the deterministic and two-stage stochastic VRP below.
\medskip
\noindent
{\bf Vehicle Routing Problem (VRP).} There is a vehicle of capacity $Q$, metric $(V,d)$ with root/depot $r\in V$ and
demands $\{q_v\le Q\}_{v\in V}$. The goal is to find a minimum cost tour of the vehicle that delivers $q_v$ units to
each $v\in V$. The demands are ``unsplittable'', i.e. the demand at any vertex must be satisfied in a single visit.
Any VRP solution corresponds to a sequence of round-trips from the depot, where at most $Q$ units of
demands are served during each round-trip. It is
well-known~\cite{AG87} that an $\alpha$-approximation ratio for TSP implies an $(\alpha+2)$-approximation algorithm for
VRP.

\medskip
\noindent
{\bf Two-stage Stochastic VRP (\svrp).} The setting is same as above, with a capacity $Q$ vehicle, metric $(V,d)$ and
depot $r\in V$. Here the demands $\{q_v\}_{v\in V}$ are random variables given by a joint demand distribution \D on
$\{0,1,\ldots,Q\}^V$, available as a black-box that can be sampled from. We are also given an inflation parameter
$\lambda\ge 1$. The goal is to compute a fixed route solution with a recourse strategy.
\begin{itemize}
\item In the first stage the algorithm computes a {\em fixed tour} $\tau$, without knowledge of the actual demand.
The tour $\tau$ consists of several round-trips from the depot: each round-trip is a cycle containing $r$ (henceforth
called $r$-tour). We represent $\tau$ as a concatenation $\{\tau_1,\ldots,\tau_F\}$ of $r$-tours. It is important to
note that $\tau$ only represents the vehicle route, and does not specify demand deliveries (this will be decided
after demand instantiations). In particular, a vertex $v$ may appear in multiple $r$-tours of $\tau$; and even if $v$
appears in $\tau$ the instantiated demand at $v$ may not eventually be satisfied by $\tau$.
\item In the second stage, the demands \oq are instantiated from \D. Knowing this, an algorithm chooses to satisfy subset $\oq_A \sse \oq$ of demands
using the fixed tour $\tau$, subject to the vehicle capacity of $Q$.
That is, for each  $r$-tour $\{\tau_i\}_{i=1}^F$ the algorithm chooses a subset $S_i\sse \tau_i$ of vertices to serve,
where $\sum_{v\in S_i}\oq_v\le Q$; and sets $\oq_A \equiv \{\oq_v : v\in \cup_{i=1}^F S_i\}$. Then the algorithm
computes a {\em recourse tour} $\sigma$
meeting all residual demands
$\oq_B=\oq\setminus \oq_A$. That is, $\sigma$ is a solution to the deterministic VRP instance with demands $\{\oq_v : v
\in V\setminus \cup_{i=1}^F S_i\}$.
\end{itemize}

Note that the demands $\oq_A$ satisfied by the fixed tour $\tau$ differs based on the instantiation \oq; however the
route taken by the vehicle stays fixed. So the first stage cost is just the length $d(\tau)$ of the fixed tour. The
recourse tour $\sigma$ clearly depends on the demand instantiation.
The second stage cost under demand \oq is $\lambda\cdot d(\sigma(\oq))$, the length of the recourse tour inflated by a
parameter $\lambda$. The objective in \svrp is to minimize the expected total cost:
$$d(\tau)\quad + \quad \lambda\,\cdot \,\Ex_{\oq\gets \D}\left[ \,d(\sigma(\oq)) \,\right]$$
For any integer $I\ge 0$, we let $[I]:=\{1,\ldots,I\}$. For a given \svrp instance, \OPT will denote its optimal value.
We let $n=|V|$ denote the number of vertices in the metric and $D=\max_{u,v}\, d(u,v)$ the diameter of the metric.

\medskip
\noindent
{\bf Our Results, Techniques and Outline.} In this paper we show:
\begin{theorem} \label{thm:algo}
There is a randomized $O(\log^{2}n \cdot \log(n\lambda))$-approximation algorithm for \svrp under arbitrary
distributions.
\end{theorem}
Using a sampling-based reduction~\cite{CCP05} we show (in Subsection~\ref{app:redn-explicit}) that the objective {\em
value} under any black-box distribution can be well-approximated by another demand distribution having support size
$m=poly(n,\lambda)$.

Then, in Section~\ref{sec:explicit} we present an $O(\log^{2}n\cdot \log(nm))$-approximation algorithm for \svrp where
$m$ is the support size of the distribution. This is a set-cover type algorithm that uses the {\em submodular
orienteering problem}~\cite{CP05,CZ05} as a subroutine. In the submodular orienteering problem there is a metric
$(V,d)$ with root $r$, length bound $B$ and monotone submodular function $f:2^V\rightarrow \mathbb{R}_+$; and the goal
is to find an $r$-tour of length at most $B$ visiting some subset $S\sse V$ of vertices so as to maximize $f(S)$.
Direct use of algorithms from~\cite{CEK06,CZ05} yields an approximation ratio worse than Theorem~\ref{thm:algo} by a
factor of $\log^\epsilon n$. Instead we give a better result for submodular orienteering on objective functions of the
type encountered in \svrp, called {\em knapsack rank-function orienteering} (\kro). In particular, we consider the {\em
ratio \kro} problem where instead of the length-bound, the objective is to maximize the ratio of function value to the
length.
\begin{theorem}\label{thm:kro}
There is a deterministic $O(\log^2 n)$-approximation algorithm for ratio knapsack rank-function orienteering.
\end{theorem}
The main idea here is to use LP rounding techniques for the related group Steiner problem~\cite{GKR00,KRS02}, augmented
with an {\em alteration} step (for the analysis). While alteration has been widely used with LP-rounding,
eg.~\cite{Srinivasan01}, we are not aware of an application in context of the group Steiner tree problem. This step
only bounds the function-value and length in expectation (separately). In order to bound their ratio, we adapt the
group Steiner derandomization from Charikar et al.~\cite{CCGG98} to our context. We defer further discussion and
details on \kro to Section~\ref{app:kro}.

Combined with the sampling-based reduction this suffices to approximate the objective {\em value} of \svrp under
black-box distributions, satisfying the guarantee in Theorem~\ref{thm:algo}. However, more work is required in order to
provide an approximate {\em solution}. This is because the recourse step in \svrp is quite non-trivial, and a solution
must specify an algorithm to construct the recourse tour for {\em any} possible demand (not merely the $m$ sampled
points). It turns out that the recourse step corresponds to solving an ``outlier'' version of VRP. Although this
problem does not admit any true approximation ratio (by a relation to generalized assignment~\cite{LST90}), in
Section~\ref{sec:black-box} we give an LP-based $O(1)$ {\em bicriteria} approximation: this suffices for
Theorem~\ref{thm:algo}.

Our second main result is a UGC-based hardness of approximation:
\begin{theorem} \label{thm:hard}
Assuming the Unique Games Conjecture, it is NP-hard to approximate \svrp to within a constant factor, even on star-like
metrics.
\end{theorem}
This is proved in Section~\ref{sec:hardness} and involves a reduction from the vertex cover problem on $k$-uniform
hypergraphs: we use a result by Bansal and Khot~\cite{BK10} which says that it is UGC-hard to distinguish between the
(yes) case when the hypergraph is almost $k$-partite and the (no) case when any vertex cover is almost the entire
vertex-set. We remark that this super-constant hardness holds in star-like metrics, where our algorithm achieves an
$O(\log (n\lambda))$-approximation. Our algorithm loses additional log-factors in going from (i) stars to trees, and
then (ii) trees to general metrics: these overheads are similar to the best known results for the related {\em group
Steiner tree} problem~\cite{GKR00}.

Finally, we consider the special case when demands are independent across vertices. Using a different algorithm we
obtain a better ratio in Section~\ref{app:indep}.
\begin{theorem}\label{thm:indep}
There is a randomized $O(\frac{\log (n\lambda)}{\log\log (n\lambda)})$-approximation algorithm for \svrp under
independent demand distributions.
\end{theorem}
We show that in this case we can enforce a certain solution structure, while losing an $O(\frac{\log
(n\lambda)}{\log\log (n\lambda)})$ factor in the optimal value. Specifically, we show that the demands can be
partitioned into two groups: one where each demand is (almost always) served by the fixed tour, and another where each
demand is served in the recourse tour.
Then we use an LP-based algorithm to find the best such partition, losing another constant factor. We leave open the
possibility of a constant approximation in the independent demands case.

\smallskip
\noindent
{\bf Related Work.} The VRP~\cite{TV01} is an extensively studied routing problem that combines aspects of both TSP and
bin-packing. Several stochastic variants of the basic problem have received attention,
eg.~\cite{SG83,B92,D05,AE07,ESU09}. Approximation algorithms for VRP with independent stochastic demands (in the a
priori model) were given in~\cite{B92,GNR12}. This paper takes a different approach, that of two-stage stochastic
optimization with recourse (along the lines of~\cite{IKMM04,RS06,SS06,GPRS11} etc). To the best of our knowledge no
prior approximation results are known for vehicle routing problems in this model.

Stochastic optimization~\cite{BL97} is a broad area dealing  with probabilistic input. Approximation algorithms for
two-stage stochastic problems were introduced by Immorlica et al.~\cite{IKMM04} and Ravi-Sinha~\cite{RS06}. Gupta et
al.~\cite{GPRS11} and Shmoys-Swamy~\cite{SS06} gave general frameworks for approximating a number of stochastic
optimization problems; the former result is combinatorial using certain cost-sharing properties, whereas the latter is
LP-based. However, these approaches do not seem directly applicable to \svrp. The results in~\cite{GPRS11,SS06} hold in
the most general distribution model, where an algorithm only receives independent samples from a black-box.  Charikar
et al.~\cite{CCP05} showed that any arbitrary distribution can be reduced to one having polynomial support (under
certain conditions). We also make use of this result in proving Theorem~\ref{thm:algo}. For most other combinatorial
optimization problems that have been considered in the two-stage stochastic model (with proportional cost inflation),
it has been observed that approximation ratios are the same order of magnitude as the underlying deterministic
problem~\cite{IKMM04,RS06,GPRS11,SS06,SS07}. A notable exception is minimum cost max-matching~\cite{KKU08}, for which
an $\Omega(\log n)$-hardness of approximation was shown. In the case of VRP %
Theorem~\ref{thm:hard} shows (under UGC) that the stochastic approximation ratio is necessarily worse than its
deterministic counterpart, even in very special metrics.

\section{Algorithm for Polynomial Scenarios}\label{sec:explicit}
Here we consider the case when the demand distribution \D is specified as a list of possible outcomes. Later on we show
how the general case of a black-box distribution can be reduced to this case. Formally \D is a multiset
$\{\oq^1,\ldots,\oq^m\}$ where the actual demand $\oq=\oq^i$ (for some $i\in[m]$) with probability $1/m$.

The main idea of our algorithm is to recast the problem as an instance of set-cover with an exponential number of sets.
Then we show that the greedy subproblem is an instance of {\em submodular orienteering} (\sop) for which a
poly-logarithmic approximation is known~\cite{CZ05,CP05}. In fact, for the type of \sop instances obtained from \svrp
we give a better approximation ratio in Section~\ref{app:kro}. Altogether, this implies Theorem~\ref{thm:algo} for
polynomial scenarios.

\medskip
\noindent
{\bf Set cover instance \I.} The groundset $U$ consists of tuples \pair{i}{v} for all scenarios $i\in[m]$ and vertices
$v\in V$, which denotes $\oq^i(v)$ demand units at $v$ under scenario $i$. For any $\pair{i}{v} \in U$ we use
$q(\pair{i}{v}) := \oq^i(v)$, and for any subset $S\sse U$, $q(S):=\sum_{t\in S} q(t)$. Instance \I has the following
two types of sets:
\begin{enumerate}
\item $S:=\cup_{i=1}^m S_i$ is a {\bf first stage set} iff $S_i\sse \{\pair{i}{v} : v\in V\}$ and $q(S_i)\le Q$ for all
$i\in[m]$. The cost of this set $S$ is the minimum length of an $r$-tour that contains all the vertices represented in
$S$.
\item For any scenario $i\in[m]$, $T\sse \{\pair{i}{v} : v\in V\}$ is a {\bf second stage set} iff $q(T)\le Q$.
The cost of set $T$ is $\lambda/m$ times the minimum length of an $r$-tour containing all vertices of $T$.
\end{enumerate}

\begin{lemma}
The set cover instance \I is equivalent to \svrp.
\end{lemma}
\begin{proof}
Recall that any feasible \svrp solution is specified by:
\begin{itemize}
\item {\em The fixed tour $\tau$.} It will be convenient to view this as a collection $\{\tau_1,\ldots,\tau_F\}$ of $r$-tours, each of which is a
 round-trip from the depot.
\item {\em For each scenario $i\in[m]$, the demands $\oq_A^i\sse \oq^i$ satisfied by the fixed tour.} Again this is viewed as follows:
for each $r$-tour $\{\tau_j\}_{j=1}^F$, $S_{i,j}\sse \{\pair{i}{v} : v\in V\}$ denotes the demands satisfied in
$\tau_j$. Note that by definition, $\bigcup_{j\in[F]} S_{i,j} \equiv  \oq_A^i$. Also due to the capacity constraint,
$q(S_{i,j})\le Q$ for each $j\in[F]$.
\item {\em For each scenario $i\in[m]$, the recourse tour $\sigma_i$ which satisfies residual demands $\oq^i\setminus \oq_A^i$.} Again
we view this as a collection $\{\sigma_{i,1},\ldots,\sigma_{i,L_i}\}$ of $r$-tours. For $k\in[L_i]$ let $T_{i,k}\sse
\{\pair{i}{v} : v\in V\}$  denote the demands satisfied in $\sigma_{i,k}$. Clearly, $\bigcup_{k\in[L_i]} T_{i,k} \equiv
\oq^i\setminus \oq_A^i$. Again $q(T_{i,k})\le Q$ for all $k\in [L_i]$.
\end{itemize}

Note that corresponding to each first stage $r$-tour $\tau_j$, the set $\bigcup_{i=1}^m S_{i,j}$ is a valid {\em first
stage set} in \I since for all $i\in[m]$ (a) $S_{i,j}\sse \{\pair{i}{v} : v\in V\}$ and (b) $q(S_{i,j})\le Q$. Moreover
the cost of this set in \I is at most $d(\tau_j)$.

Similarly, for each scenario $i\in[m]$ and second stage $r$-tour $\sigma_{i,k}$ ($k\in[L_i]$), set $T_{i,k}$ is a valid
{\em  second stage} set. The cost of this set in \I is at most $\frac{\lambda}m\cdot d(\sigma_{i,k})$.

Finally, these sets cover $U$ in \I since for each scenario $i\in [m]$, we have:
$$\left(\cup_{j=1}^F S_{i,j}\right) \,\, \bigcup\,\, \left( \cup_{k\in[L_i]} T_{i,k} \right) \quad = \quad \{\pair{i}{v} : v\in V\}$$
The total cost of this solution to \I is at most:
$$\sum_{j=1}^F d(\tau_j) \,\, + \,\, \frac{\lambda}m\cdot \sum_{i=1}^m \sum_{k=1}^{L_i} d(\sigma_{i,k}) \quad = \quad d(\tau) \,\,+ \,\,\lambda\cdot \Ex_{\oq\gets \D}\left[ d(\sigma(\oq)) \right],$$
which is just the \svrp objective value. The reverse relation (from \I to \svrp) can be shown in a similar manner, and
the lemma follows.
\end{proof}

Thus it suffices to solve the set cover instance \I. We use the greedy algorithm for set cover which requires solving
the following {\em max-coverage subproblem}: given $U'\sse U$ find a set (of either first/second type) that maximizes
the ratio of the number of $U'$-elements it covers to its cost. We give separate algorithms for this problem, under the
two types of sets.

\medskip
\noindent
{\bf Max-coverage for second stage sets.} We give a constant approximation in this case. Assume that the algorithm
knows by enumeration (i) the cost $B$ of the best ratio set (up to a factor two), and (ii) the scenario $i\in [m]$
corresponding to it. Then it suffices to find a set $T\sse U' \bigcap \{\pair{i}{v} : v\in V\}$ maximizing $|T|$ such
that $q(T)\le Q$ and cost$(T)\le B$. By the definition of second stage sets, this reduces to finding an $r$-tour
visiting the maximum vertices $W\sse\{u\in V: \pair{i}{u}\in U'\}$, having length at most $\frac{m}{\lambda}\cdot B$
{\em and} with $\sum_{u\in W} \oq^i(u)\le Q$. This is just an instance of the {\em knapsack-orienteering} problem, for
which a constant-factor approximation is known~\cite{GKNR12}.

\medskip
\noindent
{\bf Max-coverage for first stage sets.} In this case, we obtain a poly-logarithmic approximation. Again, we assume
that the algorithm knows the cost $B$ of the best ratio set (up to a factor two). Recall that unlike the previous case,
one first stage set can cover elements from several scenarios. By definition, each first stage set $S$ corresponds to
an $r$-tour visiting vertices $W\sse V$ {\em and} subsets $S_i\sse \{\pair{i}{v} : v\in W\}$ for each $i\in[m]$ such
that $\{q(S_i)\le Q\}_{i=1}^m$ and $S=\bigcup_{i=1}^m S_i$. Among all first stage sets visiting a {\em fixed vertex-set
$W\sse V$},  the maximum coverage of $U'$ equals:
$$f(W) \quad := \quad \sum_{i=1}^m \,\, \max\left\{ |S_i| \,\,\, :\,\,\, S_i\sse \{u\in W: \pair{i}{u}\in U'\},\,\, \sum_{v\in S_i} \oq^i_v \le Q\right\}$$
For each $i\in [m]$ let $f_i(W)$ denote the term inside the above summation. Recall that the cost of all first stage
sets visiting vertices $W$ is the same, namely the minimum TSP on $\{r\}\cup W$. Thus the subproblem we wish to solve
is:
\begin{equation}\label{eq:greedy-1st-stage}
\max \,\, f(W) \,\, : \,\, \mbox{ there is an $r$-tour visiting $W\sse V$ of length }\le B.
\end{equation}
Recall the submodular orienteering problem (\sop) where given metric $(V,d)$ with root $r$, bound $B$ and submodular
function $g:2^V\rightarrow \mathbb{R}_+$, the goal is to find an $r$-tour visiting some subset $W\sse V$ of vertices,
having length at most $B$ that maximizes $g(W)$. If $f$ were submodular then we can use the algorithm~\cite{CZ05,CP05}
to solve this. But $f$ is not submodular. Still, we show below that it can be well approximated by a submodular
function $g$.

We approximate each $f_i$ (point-wise) by a submodular function $g_i$. Let $V_i := \{u\in V: \pair{i}{u}\in U'\}$
denote the vertices appearing with scenario $i$ in $U'$. Define:
$$ g_i(W)\quad := \quad \max \left\{ \, \sum_{v\in V_i\cap W} x_v \,:\, \sum_{v\in W} \oq^i_v\cdot x_v \le Q,\,\,\, 0\le x_v \leq 1, \,\, \forall \,v\in
W\right\}$$ Observe that $g_i(W)$ is just an LP relaxation for a maximization $\{0,1\}$-knapsack problem. So its value
is given by the greedy algorithm that increases $x_v$ (up to $1$) in increasing order of $\{\oq^i_v : v\in V_i\cap
W\}$. On the other hand, $f_i(W)$ is the value of the same {\em integral} knapsack problem. Now, function $g_i$ can be
rewritten as the rank function of a {\em polymatroid}~\cite{Schrijver} which is submodular; see eg.~\cite{DGV08}.
Moreover, the integrality gap of the natural LP for max-knapsack is two. Thus,
\begin{cl}
$g_i$ is monotone submodular and $\frac{g_i(W)}2 \le f_i(W) \le g_i(W),\,\forall W\sse V$.
\end{cl}

So if we define $g(W) := \sum_{i=1}^m g_i(W)$ then it is submodular and maximizing $g$ in~\eqref{eq:greedy-1st-stage}
is equivalent to maximizing $f$ (up to factor two). Hence, assuming a $\rho$-approximation algorithm for \sop, we
obtain a $2\rho$-approximation algorithm for~\eqref{eq:greedy-1st-stage}. This suffices to give an
$O(\rho)$-approximation for the max-coverage subproblem. We have $\rho=O(\log^{2+\epsilon} n)$ in polynomial time using
the bicriteria approximation in Calinescu-Zelikovsky~\cite{CZ05}, and $\rho=O(\log n)$ in {\em quasi-polynomial} time
using the true approximation in Chekuri-Pal~\cite{CP05}. In Section~\ref{app:kro} we directly consider the {\em ratio}
objective corresponding to~\eqref{eq:greedy-1st-stage}, called {\em ratio knapsack rank-function orienteering}, i.e.
$$\max\left\{ \frac{f(V(\tau))}{d(\tau)} \quad : \quad \tau \mbox{ is an $r$-tour visiting vertices
}V(\tau)\,\right\},$$ and give an improved polynomial time $O(\log^2 n)$-approximation algorithm for it.

\medskip

Finally, we lose an additional $\log |U|=O(\log (mn))$ factor to solve the set cover instance \I (which is equivalent
to \svrp). Thus we obtain:
\begin{theorem}\label{thm:explicit} There is a polynomial time $O(\log^{2}n\cdot \log(nm))$-approximation algorithm
for \svrp for a polynomial number $m$ of scenarios and $n$ vertices. This ratio
improves to $O(\log n\cdot \log(nm))$ in quasi-polynomial time.
\end{theorem}
\section{Algorithm for General Distributions}\label{sec:black-box}
In this section we prove Theorem~\ref{thm:algo} under an arbitrary distribution \D that is accessed by sampling. We
denote the input \svrp instance by \J. In Subsection~\ref{app:redn-explicit} we apply a sampling-based reduction
from~\cite{CCP05} to obtain an equivalent \svrp instance $\J'$ with $m=poly(n,\lambda)$ scenarios.
This allows us to apply the algorithm from the previous section to approximate the {\em optimal value} of instance \J.
However a solution to \J must also specify a valid recourse strategy for {\em every} outcome $\oq\in\D$, and not just
for the $m$ outcomes in instance $\J'$. It turns out that the recourse step is captured by an ``outlier'' version of
VRP, and we give an LP-based constant-factor bicriteria approximation for it in Subsection~\ref{subsec:recourse}.

\subsection{Sampling Based Reduction to Polynomial Scenarios}\label{app:redn-explicit} Here we show that
sampling can be used to reduce an arbitrary demand distribution to one having a polynomial number of scenarios.

Given a fixed tour $\tau$ and scenario $\oq\in\D$, let $h(\tau,\oq)$ denote the minimum cost of a recourse tour. Note
that computing $h(\tau,\oq)$ involves choosing a subset $\oq_A$ of \oq to be served by $\tau$ (at zero cost, but
subject to capacity) and then optimally solving the VRP instance with demands $\oq-\oq_A$ and lengths inflated by
factor $\lambda$. Thus we can express the minimum objective value for a given fixed tour $\tau$ as:
$$\obj(\tau)\quad := \quad d(\tau) \,\,+\,\, \Ex_{\oq\gets \D}\left[ h(\tau,\oq)\right].$$

The optimal value of \svrp instance \J is then $\min_{\tau\in X}\, \obj(\tau)$, where $X$ denotes the set of all
possible fixed tours.
Consider drawing $m$ independent samples $\{\oq^1,\ldots,\oq^m\}$ from \D, and let $\J'$ denote the (random) instance
of \svrp with these as explicit scenarios. Define:
$$\sobj(\tau)\quad := \quad d(\tau) \,\,+\,\, \frac1m \cdot \sum_{i=1}^m h(\tau,\oq^i), \qquad \forall \mbox{ fixed tour }\tau\in X.$$

It is clear that $\min_{\tau\in X}\, \sobj(\tau)$ is the optimal value of $\J'$. We now use the result of Charikar et
al.~\cite{CCP05} to relate these two instances. For completeness we give a proof adapted to our context. Let
$D=\max_{u,v}\,d(u,v)$ denote the diameter of the metric; we assume WLOG (by scaling) that all distances are integral.
\begin{theorem}[\cite{CCP05}]\label{thm:sampling}
Using $m=\Theta(\lambda^2\,D^2\,n^2\,\log|X|)$, with probability $1-o(1)$,
$$|\obj(\tau)-\sobj(\tau)|\le 1,\qquad \mbox{ for all }\tau\in X.$$
\end{theorem}
\begin{proof}
Fix any fixed tour $\tau\in X$. Define $H=\Ex_{\oq\gets \D}\,h(\tau,\oq)$ and random variables $H_i:= h(\tau,\oq^i)$
for $i\in[m]$. Clearly $\Ex H_i=H$ for all $i\in[m]$. Note that $H_i\le 2\lambda n D$: the worst case recourse action
involves a separate round-trip to each vertex. Since $H_i/(2\lambda n D)$ are independent $[0,1]$ random variables, by
Chernoff bound~\cite{MR95},
$$\Pr \left[ \bigg| \frac1m\cdot \sum_{i=1}^m H_i  - H \bigg| \, >\, \epsilon\right] \quad \le \quad 2\exp\left(-\frac{\epsilon^2\cdot
m}{4\lambda^2n^2 D^2}\right),\qquad \forall \epsilon>0.$$ Now observe that $\obj(\tau)=d(\tau)+H$ whereas
$\sobj(\tau)=d(\tau)+\frac1m\cdot \sum_{i=1}^m H_i$. Thus using the above inequality with  $\epsilon=1$ and
$m=8\lambda^2n^2 D^2\cdot \log |X|$, we obtain:
$$\Pr\left[|\obj(\tau)-\sobj(\tau)|> 1\right] \quad \le \quad \frac2{|X|^2},\qquad \forall \tau\in X.$$
Finally, a union bound over all $X$ implies the theorem.
\end{proof}

We now show that $m=poly(n,\lambda)$.

\begin{cl}\label{cl:number-base-tour}
WLOG the number of $r$-tours in any fixed tour is at most $n$. Hence the number of edges used in the fixed tour is at
most $n^2$, and $|X|\le 2^{n^2}$.
\end{cl}
\begin{proof}
Consider an arbitrary fixed tour $\tau$ consisting of $r$-tours $\tau_1,\ldots,\tau_F$. Suppose that $F>n$: then we
will show there exists another fixed tour $\tau'$ with at most $n$ $r$-tours such that $\obj(\tau')\le \obj(\tau)$. Let
us number vertices so that the depot is numbered $0$ and $d(0,1)\le d(0,2)\le\cdots\le d(0,n)$. For each $j\in[F]$ let
$M(j)\in[n]$ denote the maximum numbered vertex in $r$-tour $\tau_j$; note $d(\tau_j)\ge 2\cdot d(0,M(j))$. Choose
$k\in[n]$ as the minimum value so that $|M^{-1}(\{k,\ldots,n\})|\le n-k$; if there is no such value then set $k=n+1$.

Let $G = M^{-1}(\{k,\ldots,n\}) \sse [F]$; note that $G=\emptyset$ when $k=n+1$. By choice of $k$, we have $|G|\le
n-k+1$.

Also by choice of $k$, it follows that $|M^{-1}(\{t,\ldots,n\})|\ge  n-t+1$ for all $t\le k-1$. Thus there is an
injective map $\phi:[k-1]\rightarrow [F]\setminus G$ such that $M(\phi(v))\ge v$ for all $v\in [k-1]$; this can be
obtained say by a greedy assignment starting from $k-1$ (recall $|G|\le n-k+1$). Due to the vertex numbering (and
definitions of $\phi,M$) we have $2\cdot d(0,v)\le 2\cdot d(0,M\circ\phi(v)) \le d(\tau_{\phi(v)})$ for all $v\in
[k-1]$. And since $\phi$ is injective, $2\cdot \sum_{v=1}^{k-1} d(0,v)\,\le \,\sum_{j\in[F]\setminus G} d(\tau_j)$.

Set $\tau'$ to consist of the following $r$-tours: (a) all singleton $r$-tours $\langle 0,v,0\rangle$ for $v\in[k-1]$,
and (b) $\{\tau_j:j\in G\}$. Using the above inequality, $d(\tau')\le d(\tau)$. Observe that vertices
$\{1,\ldots,k-1\}$ will play no role in the second stage under $\tau'$, since they are already individually covered in
a first-stage $r$-tour. Moreover, for any vertex $v\in \{k,\ldots,n\}$, the set of $r$-tours containing $v$ is
identical in both $\tau$ and $\tau'$. Hence for {\em any scenario} $\oq\in \D$, the recourse action (for vertices
$\{k,\ldots,n\}$) under $\tau$ is also feasible under $\tau'$. This implies $h(\tau',\oq)\le h(\tau,\oq)$ for all
$\oq\in\D$; and so $\obj(\tau')\le \obj(\tau)$. Also by construction, the number of $r$-tours in $\tau'$ is $k-1+|G|\le
n$.\end{proof}

\begin{cl}\label{cl:recourse-edges}
WLOG, the number of edges used in any recourse tour is at most $2n$.
\end{cl}
\begin{proof}
Note that any recourse tour (under any outcome \oq) is a solution to some deterministic VRP instance. Since there are
at most $n$ demands, and we consider unsplittable routing, it is clear that the number of edges used is at most $2n$.
\end{proof}

So far we have shown $|X|\le 2^{n^2}$. Next we will show that $D=poly(n)$, which suffices to prove $m=poly(n,\lambda)$
in Theorem~\ref{thm:sampling}.

Assume (by enumeration) that the algorithm knows an upper bound $B$ on the optimal value of instance \J (up to factor
two), i.e. $B/2\le \OPT(\J)\le B$. Let $U\sse V$ denote the vertices at distance at most $B$ from $r$. Clearly, the
optimal fixed tour does not visit any vertex outside $U$ (otherwise it incurs cost larger than $2B$). So we may always
defer demands at $V\setminus U$ to the second stage (which is what \OPT does). And by using the $O(1)$-approximation
algorithm~\cite{AG87} for deterministic VRP to serve $V\setminus U$, the cost incurred by our algorithm on $V\setminus
U$ is at most $O(1)\cdot B$. Now we can focus on the \svrp instance restricted to vertices $U$.

Consider the following modification to the metric $(U,d)$ : contract all edges of length smaller than $B/n^3$ and let
$(W,\ell)$ denote the resulting metric of shortest-path distances. We consider the natural \svrp instance $\J''$ on
$(W,\ell)$ (where \D induces the demand distribution on $W$ as well). The useful property of metric $\ell$ is that it
has maximum distance $\le 2B$ and minimum distance $\ge B/n^3$; so by scaling we obtain that it has diameter at most
$O(n^3)$. Thus we can apply Theorem~\ref{thm:sampling} to instance $\J''$ and $m=poly(n,\lambda)$ samples would
suffice. The following lemma relates the $\J''$ to the original instance \J.
\begin{lemma}\label{lem:diameter}
The optimal value $\OPT(\J'')\le B$. Moreover, any solution to $\J''$ yields a solution to \J with at most a constant
factor increase in objective.
\end{lemma}
\begin{proof}
The first part of the claim is trivial since $\J''$ is obtained from \J by contracting the metric: so $\OPT(\J'')\le
\OPT(\J)\le B$. For the other direction, consider any solution to $\J''$:
we now describe the solution corresponding to this in \J. Note that each vertex $w\in W$ corresponds to some subset
$U_w\sse U$ such that there is a spanning tree on $U_w$  in metric $d$ with each edge of length at most $B/n^3$. So
whenever vertex $w\in W$ is visited in $\J''$, we will visit all vertices of $U_w$ along an Euler tour of $MST_d(U_w)$:
this results in a cost increase of at most $2|U_w|\cdot B/n^3=O(B/n^2)$. Note also that each edge $e$ in metric
$(W,\ell)$ corresponds to some path in metric $(U,d)$ of length at most $\ell_e + n\cdot B/n^3$: so each edge traversal
causes a cost increase of at most $B/n^2$. By Claim~\ref{cl:number-base-tour} the cost increase in the fixed tour is at
most $n^2\cdot O(B/n^2)=O(B)$. Similarly, using Claim~\ref{cl:recourse-edges} the cost increase in the recourse tour
(under {\em any} outcome) is at most $2n \cdot O(B/n^2)\le O(B)$; so the increase in the expected cost of the recourse
tour is also $O(B)$. Thus any $\J''$-solution corresponds to a $\J$-solution where the increase in objective is at most
$O(B)=O(1)\cdot \OPT(\J)$.
\end{proof}

Algorithm~\ref{alg:svrp-blackbox} summarizes the \svrp algorithm.
\begin{algorithm}[h] \caption{Algorithm for \svrp under black-box distributions\label{alg:svrp-blackbox} }
{\bf Input:} \svrp instance $\J = \langle (V,d),\, r,\, Q,\, \lambda,\, \D\rangle$.

\begin{algorithmic}[1]
\STATE Guess (by enumeration) value $B$ such that $B/2\le \OPT(\J)\le B$.
\STATE Restrict instance to vertices $U=\{v\in V : d(r,v)\le B\}$. Vertices $V\setminus U$ are handled separately,
always in the recourse tour (costs $O(B)$ in expectation).
\STATE Modify metric $(U,d)$ to $(W,\ell)$ by contracting edges shorter than $B/n^3$ and recomputing shortest paths. By
scaling, $D=diameter(W,\ell)\le O(n^3)$. $\J''$ is the induced \svrp instance on $(W,\ell)$.
\STATE Apply Theorem~\ref{thm:sampling} to $\J''$ to obtain a (random) instance $\J'$ of \svrp with explicit scenarios
and $m=poly(n,\lambda)$.
\STATE Obtain fixed tour $\tau$ for $\J'$ using the algorithm in Section~\ref{sec:explicit}. Theorem~\ref{thm:sampling}
implies that this is a fixed tour for $\J''$ with increase in objective being at most one w.h.p.
\STATE By Lemma~\ref{lem:diameter}, $\tau$ is also a fixed tour for \J.
\STATE {\bf Output} the fixed tour  for \J containing  {\em five} copies of each $r$-tour in $\tau$.
\STATE {\bf Output} recourse action (for any \oq) as $\aug(\tau,\oq)$ given in Algorithm~\ref{alg:recourse}.
\end{algorithmic}
\end{algorithm}

\subsection{Specifying Recourse Actions}\label{subsec:recourse}
The recourse strategy involves the following {\em outlier VRP} problem: given a fixed tour $\tau$ (as collection
$\{\tau_1,\ldots,\tau_F\}$ of $r$-tours) and outcome $\oq\in \{0,\ldots,Q\}^V$, find
\begin{itemize}
\item a subset of vertices whose demands $\oq_A\sse \oq$ can be served by the existing route $\tau$, subject to the capacity constraint
of $Q$ on its $r$-tours; {\em and }
\item a minimum cost VRP solution to the residual demands $\oq-\oq_A$.
\end{itemize}
The optimal value of this instance is exactly function $h(\tau,\oq)$ defined Subsection~\ref{app:redn-explicit}. We
remark that when the capacity $Q=1$, the outlier VRP problem can be solved exactly using a minimum cost flow
formulation. When the fixed tour $\tau=\emptyset$ we obtain the usual VRP, which is NP-hard for $Q\ge 3$. Another
special case of outlier VRP is the {\em restricted assignment} problem~\cite{LST90}. This occurs when $V$ denotes the
set of jobs with sizes $\oq$, there are $F$ machines, and potential job-machine assignments are given by $\tau$ (job
$v$ can be assigned machine $j$ iff $v\in \tau_j$); there is an assignment of makespan $Q$ iff the outlier VRP optimum
is zero. So it is NP-hard to obtain any true approximation ratio for outlier VRP.
Instead we give an $(O(1),O(1))$ {\em bicriteria} approximation algorithm, which suffices to obtain an algorithm for
\svrp with only constant-factor increase over Theorem~\ref{thm:explicit}.

The algorithm is based on a natural LP relaxation to outlier VRP. Consider a solution with $S\sse V$ as the vertices
chosen to be served by $\tau$. Then:
\begin{itemize}
\item There is an assignment $\phi:S\rightarrow [F]$ such that (1) $v\in \tau_{\phi(v)}$ for all $v\in S$;
and (2) for each $r$-tour $j\in[F]$, the total demand assigned to it $\sum_{v\in \phi^{-1}(j)} \oq_v \le Q$.
\item The objective value is the optimum VRP on metric $(V,\, d)$, depot $r$,
capacity $Q$ and demands $\{\oq_v : v\in V\setminus S\}$. Using known lower-bounds for VRP~\cite{HK85,AG87}, at the
loss of a constant factor, this is just $\mst(V\setminus S) + \flow(V\setminus S)$ where for any $T\sse V$, $\mst(T)=$
length of minimum spanning tree on $\{r\}\bigcup T$, and $\flow(T):=\frac1Q \sum_{v\in T} \oq_v\cdot d(r,v)$.
\end{itemize}

Thus we can write the following integer programming formulation for outlier VRP, at the loss of $O(1)$-factor.
\begin{alignat}{10}
&\min \quad \sum_{e\in E} d_e\cdot z_e \,\, + \,\, \frac1Q\sum_{v\in V} d(r,v)\cdot \oq_v\cdot (1-x_v) \label{eq:lp-obj}\\
&\mbox{s.t.} \quad \sum_{v\in \tau_j} \oq_v\cdot y_{v,j} \leq Q  \quad \qquad  \forall j\in[F], \label{eq:lp-capacity}\\
&\quad \quad \sum_{j\in[F] : v\in\tau_j} y_{v,j} = x_v \quad \quad \quad \forall v\in V, \label{eq:lp-assign}\\
&\quad \sum_{e\in \delta(U)} z_e \ge 1-x_v \quad \quad \quad \forall U\not\ni r,\,\, \forall v\in U, \label{eq:lp-cut}\\
&\quad x_v, y_{v,j} \in \{0,1\} \qquad \qquad  \forall v\in V,\,\, \forall j\in[F], \notag \\
&\quad z_e\ge 0 \qquad \qquad \qquad  \quad \forall e\in E.\notag
\end{alignat}
Above $x_v$ is one iff $v\in S$, i.e. served by $\tau$. Variables $y_{v,j}$ denote the assignment $\phi:S\rightarrow
[F]$. Constraint~\eqref{eq:lp-assign} ensures that each $v\in S$ is assigned to some $\phi(v)$ such that $v\in
\tau_{\phi(v)}$. Constraint~\eqref{eq:lp-capacity} enforces the total assignment to each $r$-tour is at most $Q$. Also
$E={V \choose 2}$ denotes the edge-set of the metric, and for any $U\sse V$, $\delta(U)$ denotes the edges with exactly
one vertex in $U$. Constraint~\eqref{eq:lp-cut} says that $\{z_e:e\in E\}$ is a fractional spanning tree connecting the
vertices $\{v: x_v=0\}=V\setminus S$ to $r$. In the objective~\eqref{eq:lp-obj}, the first term is the length of the
fractional spanning tree (corresponding to $\mst(V\setminus S)$), and the second term is $\flow(V\setminus S)$.
Dropping the integrality gives us an LP relaxation $\lp(\tau,\oq)$. We can solve this LP in polynomial time, and next
we describe a rounding algorithm.
\begin{algorithm}[h] \caption{Algorithm for computing recourse action $\aug(\tau,\oq)$.\label{alg:recourse} }
{\bf Input:} $(V, d)$, $r$, capacity $Q$, fixed tour $\tau=\{\tau_1,\ldots ,\tau_F\}$, outcome $\oq$.
\begin{algorithmic}[1]
\STATE Let $(x,y,z)$ denote the optimal LP solution.
\STATE Set $S\gets \{v\in V: x_v\ge \frac12\}$.
\STATE \label{step:restr-assign} Consider the following instance of {\em restricted assignment}
$$\begin{array}{ll}
\sum_{v\in \tau_j} \oq_v\cdot \alpha_{v,j} \,\, \leq \,\, 2\cdot Q  & \quad \forall j\in[F], \\
\sum_{j\in[F] : v\in\tau_j} \alpha_{v,j} \,\, = \,\, 1 & \quad \forall v\in S, \\
0\,\, \le \,\, \alpha_{v,j} \,\, \le \,\, 1& \quad \forall v\in S,\,\, \forall j\in[F]. \\
\end{array}$$
\STATE By the rounding algorithm in~\cite{LST90}
we can obtain an integral assignment
$\phi:S\rightarrow [F]$ such that $v\in\tau_{\phi(v)}$ for all $v\in S$ and $\oq\left(\phi^{-1}(j)\right)\le 3Q$ for
all $j\in[F]$.
\STATE For each $j\in[F]$, partition $\phi^{-1}(j)=\bigsqcup_{l=1}^5 S_{j,l}$ into at most five parts such that
$\{q(S_{j,l})\le Q\}_{l=1}^5$. This can be done by a greedy algorithm.
\STATE {\bf Output} for each $j\in[F]$ and $l\in\{1,\ldots,5\}$, vertices $S_{j,l}$ as served by
$\tau_j$.
\STATE {\bf Output} an $O(1)$-approximate VRP solution~\cite{AG87} on vertices $V\setminus S$ as recourse tour.
\end{algorithmic}
\end{algorithm}

Observe that the solution from Algorithm~\ref{alg:recourse} uses {\em five} copies of the fixed tour $\tau$, whereas we
will bound the cost against $\lp(\tau,\oq)$.

First we show that our assignment $S$ to the fixed tour is indeed feasible (using $5$ copies). It is clear that setting
$\alpha_{v,j}=2\,y_{v,j}$ for $v\in S,\,j\in[F]$ gives a feasible fractional solution to the restricted assignment
instance in Step~\ref{step:restr-assign} of Algorithm~\ref{alg:recourse}: this follows from
constraints~\eqref{eq:lp-capacity} and~\eqref{eq:lp-assign} using $\{x_v\ge \frac12\}_{v\in S}$. Thus the rounding
algorithm from~\cite{LST90} can be employed on $\alpha$ to obtain an integral solution $\phi$ having load at most
$2Q+\max_v \oq_v\le 3Q$ for each $j\in[F]$. Then for each $j\in[F]$, we partition $\phi^{-1}(j)$ starting with the
trivial partition into singletons and greedily merging parts as long as each part $\le Q$: this results in at most $5$
parts.
Thus vertices $S$ can be feasibly assigned to $5\,\tau$.

Next we bound the cost of our recourse tour by $O(1)\cdot \lp(\tau,\oq)$. Observe that $x_v<\frac12$ for each $v\in
V\setminus S$: so constraint~\eqref{eq:lp-cut} implies that $2\cdot z$ is a fractional spanning tree on $\{r\}\cup
(V\setminus S)$. Hence $\mst(V\setminus S)\le 4\,(\mathbf{d}\cdot\mathbf{z})$. Moreover, it is clear that
$\flow(V\setminus S)\le 2\cdot \frac1Q\sum_{v\in V} d(r,v)\cdot \oq_v\cdot (1-x_v)$. Thus the LP objective:
$$\lp(\tau,\oq)\quad \ge \quad \frac14\cdot \mst(V\setminus S) + \frac12\cdot \flow(V\setminus S).$$
Since the VRP algorithm (on demands $V\setminus S$) achieves a constant approximation relative to these lower bounds,
it follows that the recourse cost is $O(1)\cdot \lp(\tau,\oq)$.
\begin{theorem}\label{thm:recourse}
There is an $(O(1),\,5)$-bicriteria approximation algorithm for outlier VRP, that uses the fixed tour at most five
times.
\end{theorem}

\section{Algorithm for Ratio Knapsack Rank-function Orienteering }\label{app:kro}
In this section we give an improved result for the ratio version of submodular orienteering, when the objective is a
sum of ``knapsack rank-functions''. This can be used as a subroutine for \svrp to yield Theorem~\ref{thm:algo}.

An instance of the {\em knapsack rank-function orienteering} problem (\kro) consists of metric $(V,d)$ with root $r$
and length bound $B$. The objective is a sum of $m$ knapsack rank-functions $f_1,\ldots,f_m:2^V \rightarrow
\mathbb{R}_+$. For a solution visiting vertices $U$, the objective value is $\sum_{i=1}^m f_i(U)$. The goal is to find
an $r$-tour of length at most $B$ having maximum objective. Each knapsack rank-function $f_i$ is:
$$f_i(U) \quad := \quad \max\left\{ \,\sum_{v\in S} w^i_v \,\,\, :\,\,\, S\sse U,\,\, \sum_{v\in S} c^i_v \le 1\,\right\},\quad \forall \,U\sse V$$
Above $w^i:V\rightarrow \mathbb{R}_+$ and $c^i:V\rightarrow [0,1]$ denote profits and sizes at the vertices; so
$f_i(U)$ is the maximum profit in any subset of $U$ having size at most one. (Although a knapsack rank-function may not
be submodular, it can always be approximated within factor two by a submodular function, as discussed in the end of
Section~\ref{sec:explicit}.)

Here we consider the ratio version of this problem, called {\em ratio \kro}, where given metric $(V,d)$ with root $r$
and knapsack rank-functions $f_1,\ldots,f_m$, the goal is:
$$\max\,\,\left\{ \,\,\frac{\sum_{i=1}^m f_i(V(\tau))}{d(\tau)} \,: \, \tau \mbox{ is an $r$-tour }\right\}$$
Above $V(\tau)\sse V$ are the vertices visited by $\tau$ and $d(\tau)$ is its length.

Note that the max-coverage problem for first stage sets corresponds exactly to ratio \kro. Here we obtain an $O(\log^2
n)$-approximation algorithm for ratio \kro, which combined with the algorithm in Section~\ref{sec:explicit} implies
Theorem~\ref{thm:algo}.

\smallskip
\noindent
{\bf Related Work.} \kro is closely related to the group Steiner tree problem, where given metric $(V,d)$ with root $r$
and $m$ groups $G_1,\ldots,G_m\sse V$ of vertices, the goal is to find a minimum length tree connecting $r$ to at least
one vertex of each group. The best known approximation ratio for this problem is $O(\log^2n\cdot\log m)$~\cite{GKR00}.
Formally, ratio \kro generalizes the ``density group Steiner'' problem~\cite{CCGG98}, which involves finding an
$r$-tour maximizing the ratio of number of groups covered to its length: setting $w^i_v=c^i_v=\textbf{1}[v \in G_i]$ in
ratio \kro gives us density group Steiner.
Our algorithm builds on~\cite{CCGG98}; however the previous algorithm is not directly applicable since the natural LP
relaxation to ratio \kro seems weak. We strengthen the LP relaxation for \kro by adding extra constraints (see
constraints~\eqref{LP:kro4} below), that are motivated from the related {\em covering} Steiner tree
problem~\cite{KRS02}.
Moreover, as with all LP-based algorithms for group Steiner type problems, the main rounding step is the dependent
randomized rounding (called GKR rounding below) from~\cite{GKR00}. Even with a good LP relaxation and rounding step,
previously used analysis such as~\cite{GKR00,KRS02,GS06} is inadequate for bounding the profit in ratio \kro as shown
next.

\noindent
{\em Example 1:} Consider a star-like tree with a single edge $(r,u)$ from the root and $t$ other edges
$(u,v_1),\ldots,(u,v_t)$, all of unit length. There is a single knapsack with zero sizes ($\mathbf{c}=\mathbf{0}$), and
profits $\mathbf{w}$  of one at each $V'=\{v_1,\ldots,v_t\}$ (and zero at $r,u$). Suppose the LP solution sets value
$x=1/t^2$ for all edges. The fractional profit is $\mu=1/t$. The analysis in all of~\cite{GKR00,KRS02,GS06} attempts to
upper bound:
\begin{equation}\label{eq:gkr-prob} \Pr[\mbox{number of $V'$-vertices chosen in GKR rounding } < \mu/2].\end{equation}
In this example, the solution (from GKR rounding) is the entire tree with probability $\frac1{t^2}$, and empty
otherwise. So the probability~\eqref{eq:gkr-prob} is $1-\frac1{t^2}$, which by itself would only imply an expected
profit of $1/t^2 \ll \mu$ (although the actual expected profit is $\frac1t$). Such an analysis is sufficient when
$\mu=\Omega(1)$ as in~\cite{GKR00,KRS02,GS06}, but we are not guaranteed this in ratio \kro.

Instead of using a bound on the probability~\eqref{eq:gkr-prob}, we directly lower bound the {\em expected} profit
using a different analysis. In particular our main idea is to use an {\em alteration} step after GKR rounding (see
Lemma~\ref{lem:kro-main}). While alteration has been used with LP-rounding before, eg.~\cite{Srinivasan01}, we are not
aware of an application in context of the group Steiner tree problem; moreover we only use alteration in our analysis
and not in the algorithm.

In the next Subsection~\ref{subsec:kro-lp}, we present the LP relaxation that we use. In
Subsection~\ref{subsec:kro-expectation} we show that the GKR rounding ensures (i) high expected profit and (ii) low
expected length (individually). Then in Subsection~\ref{subsec:kro-derand} we use the derandomization of GKR
rounding~\cite{CCGG98}, and show that it leads to a single (deterministic) solution having a high profit/length ratio.
Altogether we obtain an $O(\log^2n)$-approximation algorithm for ratio \kro.

\subsection{LP relaxation}\label{subsec:kro-lp}
At the loss of an $O(\log n)$ factor in the approximation ratio, we can assume that the metric is a tree $T$ (with edge set $E(T)$) rooted at
$r$ having $\ell=O(\log n)$ levels~\cite{FRT04}. We also enumerate over all choices of the length $B$ of the optimal
ratio \kro solution.\footnote{It suffices to know the length up to a constant factor; so there are only polynomially
many choices.} Then we use an LP relaxation similar to the LP for group Steiner tree~\cite{GKR00}.

First some notation: For any edge $e$ in $T$, we denote by $\pi(e)$ its parent edge. Similarly $\pi(v)$ is the parent
edge of any vertex $v\in V$. For root edges $e$, the $x_{\pi(e)}$ values are
fixed to 1, since the root is always part of the solution. For any $e\in E(T)$, the subtree below edge $e$ is denoted $T_e$.
\begin{alignat}{10}
LP(B)\quad=\qquad &\max \quad \sum_{i=1}^m \sum_{v\in V} w^i_v\cdot z^i_v \\
&\mbox{s.t.} \qquad x_{\pi(e)} \ge x_e, \qquad \forall \, e\in T  \label{LP:kro1}\\
&\qquad z^i_v \le x_{\pi(v)}, \qquad \forall \, v\in V,\, i\in[m]\label{LP:kro2}\\
&\qquad \sum_{v\in V} c^i_v\cdot z^i_v \le 1, \qquad \forall \, i\in[m]\label{LP:kro3}\\
&\qquad \sum_{v\in V(T_e)} c^i_v\cdot z^i_v \le x_e, \qquad \forall \, e\in T,\, i\in[m]\label{LP:kro4}\\
&\qquad  \sum_{e\in E(T)} d_e\cdot x_e \le \frac{B}2 \label{LP:kro5}\\
&\qquad  \textbf{0}\le \textbf{x},\textbf{z}\le \textbf{1} \notag
\end{alignat}
Let us show that restricting $x$ and $z$ to integer values gives a valid formulation of \kro. In the intended integral
solution, $x_e$ is an indicator denoting whether/not edge $e$ is chosen. Constraints~\eqref{LP:kro1} ensure
monotonicity, that the solution is a subtree rooted at $r$. Constraints~\eqref{LP:kro5} bound the length of the subtree
by $B/2$ (so the corresponding Euler tour has length at most $B$ as required). Vertex $v$ is visited by the solution
iff $x_{\pi(v)}=1$; let $U = \{v\in V\,:\, x_{\pi(v)}=1\}$ denote the vertices visited. For each knapsack rank function
$i\in [m]$, variables $z^i$ denote its maximum profit subset $S_i= \{v\in V\,:\, z^i_{v}=1\}$. By~\eqref{LP:kro2},
$S_i\sse U$ as required. Moreover~\eqref{LP:kro3} ensures that the total size of $S_i$ (in the $i^{th}$ knapsack) is at
most one. Finally, the objective~\eqref{LP:kro1} is the sum of profits from each knapsack.

Although we do not need constraint~\eqref{LP:kro4} to show a valid integer programming formulation, it is crucial in
the rounding step. A similar constraint was used in~\cite{KRS02} for the related covering Steiner problem. Notice that
it indeed holds for integral solutions, so the resulting LP is a valid relaxation of \kro. For a fractional solution,
\eqref{LP:kro4} says that even {\em conditional} on edge $e$ being chosen, the total size (in knapsack $i$) from
subtree $T_e$ is at most one.

\smallskip
\noindent
{\bf Algorithm Overview.} For each estimate $B$, we solve the above $LP(B)$ and apply the deterministic rounding
algorithm in Subsection~\ref{subsec:kro-derand}, which guarantees a solution having profit/length ratio at least
$\Omega(\frac1\ell)\cdot \frac{LP(B)}{B}$. Finally we output the best ratio solution amongst all $B$s. Note, if $B^*$
denotes the length of the optimal ratio \kro solution then $LP(B^*)/B^*$ is at least the optimum ratio. So with
$B\approx B^*$ we obtain an $O(\ell)$-approximation to ratio \kro.

\subsection{Expectation guarantee in \kro} \label{subsec:kro-expectation}
Here we show that the natural GKR rounding step produces a solution having expected length at most $B$ and expected
profit at least $\Omega(LP(B)/\ell)$. Note that this {\em does not} bound the expectation of profit/length. Still, this
property is used by the algorithm in the next subsection to produce a deterministic solution to ratio \kro having value
$\Omega(\frac1\ell)\cdot \frac{LP(B)}{B}$.

\begin{algorithm}[h] \caption{LP rounding for \kro \label{alg:kro} }
\begin{algorithmic}[1]
\STATE Solve the LP relaxation $LP(B)$ for \kro to obtain $(x,z)$.
\STATE \label{kro-step2} Perform the (dependent) rounding from~\cite{GKR00}, i.e. choose each edge $e\in E(T)$
independently with probability $\frac{x_e}{x_{\pi(e)}}$ and retain only subtree $F\sse T$ connected to $r$.

\FOR{$i\in[m]$}
\STATE \label{kro-step3} For each vertex $v\in V(F)$ choose $v$ into $S_i$ independently with probability
$\frac{z^i_v}{x_{\pi(v)}}$.
\STATE \label{kro-step4} If $c^i(S_i)> 4\ell$ then set $R_i\gets \emptyset$, else $R_i\gets S_i$.
\ENDFOR
\STATE {\bf output} $r$-tour corresponding to $F$.
\end{algorithmic}
\end{algorithm}

In the above algorithm we assume that if $e$ is a root edge then $x_{\pi(e)} = 1$,
since the root is always a part of the solution.
Steps~\ref{kro-step2} and~\ref{kro-step3} are the GKR rounding, and Step~\ref{kro-step4} is the alteration step. It is
clear that in Step~\ref{kro-step2} we have:  $\Pr[e\in F]=x_e$ for each edge $e$, and so $\Pr[v\in F]=x_{\pi(v)}$ for
each vertex $v$. Note that the expected length of $F$ is:
\begin{equation} \label{eq:kro-exp-length} \Ex\left[\sum_{f\in F} d_f\right] \quad = \quad \sum_e d_e\cdot x_e \quad \le \quad \frac{B}2\end{equation}
So taking an Euler tour of $F$, the expected solution length is at most $B$.
It remains to bound the expected profit. Notice that for each knapsack $i\in[m]$, we have $R_i\sse F$ and $c^i(R_i)\le
4\ell$ with probability one, from Step~\ref{kro-step4}. So a greedy partitioning of $R_i$ yields $8\ell$ parts each of
size at most one; and by averaging some part has profit at least $w^i(R_i)/(8\ell)$. Thus:
\begin{equation}\label{eq:kro-exp-profit1}
\Ex \left[\sum_{i=1}^m f_i(F) \right] \quad \ge \quad \frac1{8\ell} \sum_{i=1}^m  \Ex\left[ w^i\left( R_i \right)
\right] \quad \ge \quad \frac1{8\ell}\cdot \sum_{i=1}^m \sum_v w^i_v\cdot \Pr[v\in R_i]\end{equation}

We bound $\Pr[v\in R_i]$ in Lemma~\ref{lem:kro-main} below. Before doing that, we introduce some notation that will
also be useful in the subsequent derandomization step. For any $i\in [m]$ and $u,v\in V$ let $I^i_v$ denote the {\em
indicator} of the event ``$v\in S_i$''; and let $I^i_{u,v}$ denote the indicator of event ``$u\in S_i$ {\em and} $v\in
S_i$''. Also for $i\in [m]$ and $v\in V$ let $J^i_v$ indicate whether ``$v\in R_i$''. Due to Step~\ref{kro-step4} it is
clear that
\begin{equation}\label{eq:indicator-J}
J^i_v \quad \ge \quad I^i_v - \frac1{4\ell}\,\sum_{u\in V} c^i_u\cdot I^i_{u,v}\,,\qquad \forall i\in[m] \mbox{ and
}v\in V.
\end{equation}

\begin{lemma}\label{lem:kro-main}
For each $i\in[m]$ and $v\in V$, $$\Pr\left[ v\in R_i\right] \quad =\quad \Ex[J^i_v]   \quad\ge \quad \Ex[I^i_v] -
\frac1{4\ell}\,\sum_{u\in V} c^i_u\cdot \Ex[I^i_{u,v}] \quad\ge \quad\frac{z^i_v}{4}.$$
\end{lemma}
\begin{proof} Observe that $\Pr[v\in S_i]=\Ex[I^i_v]=z^i_v$. Let us now {\em condition} on $I^i_v=1$, i.e. $\{v\in S_i\}$.
Let $\langle e_1,\ldots,e_\ell\rangle$ denote the edges on the path from $r$ to $v$; clearly all
these edges are in $F$. For any $j\in\{1,\ldots ,\ell\}$ define $T'_j\sse T_{e_j}\setminus \{v\}$ to be those vertices
whose least common ancestor with vertex $v$ is edge $e_j$. Also define $T'_0$ to be the vertices whose least common
ancestor with vertex $v$ is the root $r$; set $x_{e_0}=1$. For any $j\in\{0,1,\ldots,\ell\}$ and $u\in T'_{j}$ notice
that:
$$\Ex[I^i_u \,|\, I^i_v=1\,] \quad=\quad \Pr[\, u\in S_i \,|\, v\in S_i\,] \quad = \quad \Pr[\, u\in S_i \,|\, e_j\in F\,] \quad = \quad \frac{z^i_u}{x_{e_j}}$$
Taking expectations, using~\eqref{LP:kro4} and $T'_j\sse T_{e_j}$ we have for each $j\in[\ell]$,
$$ \sum_{u\in T'_j} \, c^i_u\cdot \Ex \left[ \, I^i_u\,|\, I^i_v=1\,\right]\quad  =\quad \frac{1}{x_{e_j}} \sum_{u\in T'_j} z^i_u\,\cdot \, c^i_u
\quad \le \quad 1$$

Observe that $\cup_{j=0}^\ell T'_j = V\setminus \{v\}$. So summing the above,
\begin{equation}\label{eq:kro-condn}
\sum_{u\in V\setminus v} \, c^i_u\cdot \Ex \left[ \, I^i_u\,|\, I^i_v=1\,\right] \quad  = \quad \sum_{j=0}^\ell
\,\sum_{u\in T'_j} \, c^i_u\cdot \Ex \left[ \, I^i_u\,|\, I^i_v=1\,\right]\quad \le  \quad \ell+1
\end{equation}

By Inequality~\eqref{eq:indicator-J} which is due to the alteration Step~\ref{kro-step4},
$$ \Ex[\, J^i_v\,|\, I^i_v=1\,] \,\, \ge  \,\, 1- \frac1{4\ell}\cdot \sum_{u\in V} c^i_u\cdot \Ex[ I^i_{u,v} \,|\,
I^i_v=1] \,\, = \,\, 1- \frac1{4\ell}\cdot \sum_{u\in V} c^i_u\cdot \Ex[ I^i_{u} \,|\, I^i_v=1] \,\, \ge\,\,
1-\frac{\ell+2}{4\ell} \,\ge \,\frac14$$ The second last inequality is by~\eqref{eq:kro-condn} and the fact that
$c^i_v\le 1$; the last inequality uses $\ell\ge 1$.  The lemma now follows since $\Ex[I^i_v]=z^i_v$.
\end{proof}

Combining~\eqref{eq:kro-exp-profit1} and Lemma~\ref{lem:kro-main} we have:
\begin{equation}\label{eq:kro-exp-profit2}
\Ex \left[\sum_{i=1}^m f_i(F) \right] \quad \ge \quad \frac1{8\ell}\cdot \sum_{i=1}^m \sum_v w^i_v \cdot
\left(\Ex[I^i_v] - \frac1{4\ell}\,\sum_{u\in V} c^i_u\cdot \Ex[I^i_{u,v}]\right) \quad \ge \quad \frac1{32\ell}\cdot
\sum_{i=1}^m \sum_v w^i_v\cdot z^i_v\end{equation}

Notice that~\eqref{eq:kro-exp-length} and~\eqref{eq:kro-exp-profit2} bound the {\em expected} length and profit
respectively. This is not sufficient for the ratio \kro problem. It would suffice to bound the length and profit {\em
simultaneously} (instead of expectation). But this is not possible: Consider the instance in Example~1 and the
fractional solution with $x_e=1/t$ for all edges. This is feasible to the above LP with $B\approx 2$, and has profit
$LP(B)\ge 1$. In this example, Algorithm~\ref{alg:kro} produces the following integral solution: the entire tree
(profit = length $=t$), with probability $1/t$, and the empty tree (profit = length $=0$) otherwise. Neither of these
solutions satisfies bounds on {\em both} profit and length.

Instead, we show that one can obtain a solution of high ratio ``profit/length'' by {\em derandomizing}
Algorithm~\ref{alg:kro}. This deterministic algorithm uses pessimistic estimators and is similar to~\cite{CCGG98};
however the details are quite different since we analyze a different random process.

\subsection{Deterministic Algorithm for Ratio \kro} \label{subsec:kro-derand} For the randomized algorithm~\ref{alg:kro} recall the indicator
variables $I^i_v$s and $I^i_{u,v}$s. Also let $K_f$ for any edge $f\in E(T)$ denote the indicator of event ``$f\in
F$''. Define the following estimators for profit and length: \begin{equation}\label{eq:defn-P-D} P \,\, := \,\,
\sum_{i=1}^m \sum_{v\in V} w^i_v \cdot \left( I^i_v - \frac1{4\ell}\,\sum_{u\in V} c^i_u\cdot  I^i_{u,v} \right)\qquad
\mbox{and}\qquad D \,\, := \,\, \sum_{f\in E(T)} d_f\cdot K_f\end{equation}

We have $\Ex[P]/\Ex[D]\ge \sum_{i=1}^m \sum_v w^i_v\cdot z^i_v/(4B)=LP(B)/(4B)$ by~\eqref{eq:kro-exp-profit2}
and~\eqref{eq:kro-exp-length} which is our initial estimate of the ratio. We will inspect edges $e$ of $T$ one at a
time and decide (deterministically) whether/not to include $e$ so that the ratio estimate does not decrease. At the end
of the algorithm, we obtain a deterministic subtree $F^*$ with ratio at least $\Ex[P]/\Ex[D]$. Details now follow.

To keep notation simple, we extend tree $T$ slightly. For each $v\in V$ add leaves $\{\langle i,v\rangle : i\in[m]\}$
that are adjacent to $v$ (with zero length edges). Let $L_i = \{\langle i,v\rangle \,: \, v\in V\}$ denote the leaves
corresponding to knapsack $i\in[m]$. Define the following fractional values $y^*$ on edges according to the optimal LP
solution $(x,z)$.
$$y^*_e=\left\{\begin{array}{ll} x_e& \mbox{ if $e$ is an edge in the original tree},\\
 z^i_v & \mbox{ if $e=(v,\langle i,v\rangle)$ is a new leaf-edge.}\end{array}\right.$$

Notice that the GKR rounding steps~\ref{kro-step2} and~\ref{kro-step3} in Algorithm~\ref{alg:kro} correspond to
choosing each edge $e$ in the modified tree independently w.p. $y^*_e/y^*_{\pi(e)}$ and retaining the subtree $F$
connected to $r$. (Recall that $\pi$ maps each edge/vertex to its parent edge.) Moreover, for each $i\in[m]$ subset
$S_i\equiv V(F)\cap L_i$.

In our algorithm we will be dealing with trees ${\cal T}$ (with vertex set $V({\cal T})$ and edges set $E({\cal T})$)
derived from $T$ and edge-weights $y$ (on ${\cal T}$) derived
from $y^*$. We will always have the property that $y$ is non-increasing from the root $r$ to any leaf. For any tree
${\cal T}$ and edge-weights $y$, define: {\small \begin{equation}\label{eq:recursion-P-D} P(y,{\cal T}) \,\, :=
\,\,\sum_{i=1}^m \sum_{v\in V({\cal T})\cap L_i} w^i_v \cdot \left( y_{\pi(v)}  - \frac1{4\ell}\,\sum_{u\in V({\cal T})\cap
L_i} c^i_u\cdot \frac{y_{\pi(u)}\cdot y_{\pi(v)}}{y_{\theta(u,v)}} \right) \quad \mbox{and}\quad  D(y,{\cal T}) \,\, :=
\,\, \sum_{f\in E({\cal T})} d_f\cdot y_f\end{equation}} where $\theta(u,v)$ denotes the least-common-ancestor edge of
vertices $u$ and $v$. For ease of notation, we assume WLOG that there is always a dummy edge above the root $r$ having
$y$-value one. Observe that these values correspond precisely to the expectation of random variables $P$ and $D$
from~\eqref{eq:defn-P-D}, in tree ${\cal T}$ when  GKR rounding is performed with edge-values $y$.

\begin{algorithm}[h] \caption{Deterministic algorithm for Ratio \kro \label{alg:ratio-kro}}
\begin{algorithmic}[1]
\STATE Solve the LP relaxation for \kro to obtain $(x,z)$.
\STATE \label{step-kro-ratio1} Extend tree $T$ by adding leaves $\cup_{i=1}^m L_i$ and define $y^*$ as above.
\STATE Initialize tree $F\gets \{r\}$ and $y\gets y^*$.
\WHILE{there is edge $e\in T$ incident to $r$}
\STATE Let $T_2=T_e$ denote the subtree below $e$ and $T_1=T\setminus T_2\setminus \{e\}$; see
Figure~\ref{fig:det-tree}.
\STATE Set $T''\gets T_1\cup F$, and
$$y''_f \gets \left\{
\begin{array}{ll}
y_f& \mbox{if } f\in T_1\\
1 & \mbox{if } f\in F\end{array} \right.
$$
\STATE Compute estimates $P_0=P(T'',\, y'')$ and $D_0=D(T'',\, y'')$ upon {\em excluding} $e$
by~\eqref{eq:recursion-P-D}.
\STATE Set $T'\gets (T \mbox{ contract }e)\cup (F\cup\{e\})$, and
$$y'_f \gets \left\{
\begin{array}{ll} \frac{y_f}{y_e} &
\mbox{if } f\in T_2\cup\{e\}\\
y_f& \mbox{if } f\in T_1\\
1 & \mbox{if } f\in F\end{array} \right.
$$
\STATE Compute estimates $P_1=P(T',\, y')$ and $D_1=D(T',\, y')$ upon {\em including} $e$ by~\eqref{eq:recursion-P-D}.
\IF{$\frac{P_0}{D_0}\ge \frac{P_1}{D_1}$}
\STATE Set $T\gets T_1$ and $y\gets y''$.
\ELSE
\STATE Set $T\gets (T \mbox{ contract }e)$ and $y\gets y'$. Also $F\gets F\cup\{e\}$ and $y_e\gets 1$.
\ENDIF
\ENDWHILE
\STATE {\bf output} $r$-tour corresponding to $F$.
\end{algorithmic}
\end{algorithm}

\begin{figure}
\begin{center}\includegraphics[scale=0.65]{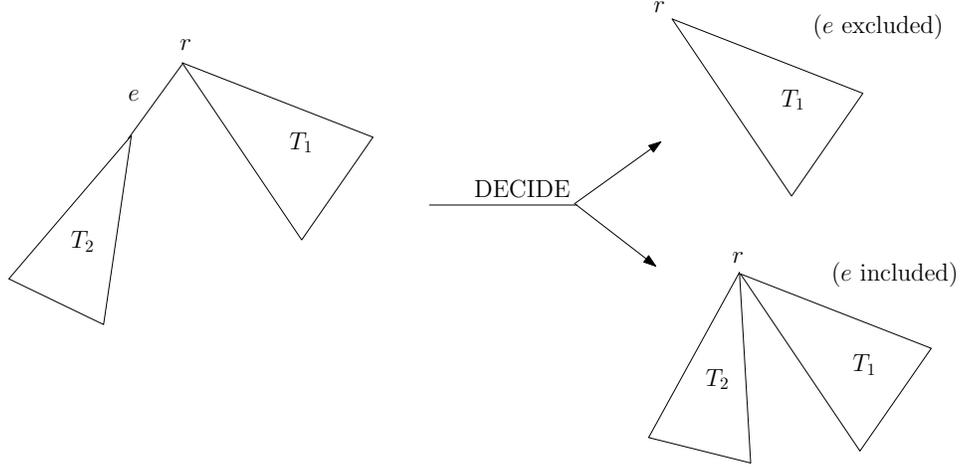}
\caption{The deterministic edge selection step. \label{fig:det-tree}}
\end{center}
\end{figure}

\begin{lemma}
At any iteration in Algorithm~\ref{alg:ratio-kro}, after Step 9, we have $P(y,T\cup F)=y_e\cdot P(y',T')+(1-y_e)\cdot P(y'',T'') =
y_e\cdot P_1 + (1 - y_e)\cdot P_0$ and
$D(y,T\cup F)=y_e\cdot D(y',T')+(1-y_e)\cdot D(y'',T'') = y_e\cdot D_1 + (1-y_e)\cdot D_0$.
\end{lemma}
\begin{proof}
Consider first the equation for $D$. We have:
$$D(y',T') = d_e + \sum_{f\in E(F)} d_f +  \sum_{f\in E(T_1)} d_f\cdot y_f+ \sum_{f\in E(T_2)} d_f\cdot \frac{y_f}{y_e}$$
$$D(y'',T'') = \sum_{f\in E(F)} d_f +  \sum_{f\in E(T_1)} d_f\cdot y_f.$$
It follows that in the convex combination $y_e\cdot D(y',T')+(1-y_e)\cdot D(y'',T'')$, each edge $f\in E(F)$ contributes
$d_f$ and each edge $g$ in $T=T_1\cup T_2\cup\{e\}$ contributes $d_g\cdot y_g$. This is exactly $D(y,T\cup F)$.

Next consider the equation for $P$. We will show equality term-by-term in the expression~\eqref{eq:recursion-P-D} for
$P$.

Consider the first summation $P_a(y,{\cal T}) := \sum_{i=1}^m \sum_{v\in V({\cal T})\cap L_i} w^i_v \cdot y_{\pi(v)}$. We
show that the contribution of any term ``$i\in [m]$ and $v\in L_i$'' to $P_a(y,T\cup F)$ is the same as to $y_e\cdot
P_a(y',T')+(1-y_e)\cdot P_a(y'',T'')$. %
\begin{center}
{\renewcommand{\arraystretch}{1.5}
\renewcommand{\tabcolsep}{0.2cm}\begin{tabular}{|c|c|c|c|c|}
\hline  Cases & $P_a(y,T\cup F)$ & $P_a(y',T')$ & $P_a(y'',T'')$ & $y_e\cdot P(y',T')+(1-y_e)\cdot P(y'',T'')$\\
\hline $\pi(v)\in E(F)$ & $w^i_v$ & $w^i_v$ & $w^i_v$ & $w^i_v$ \\
\hline $\pi(v)\in E(T_1)$ & $w^i_v\cdot y_{\pi(v)}$ & $w^i_v\cdot y_{\pi(v)}$ & $w^i_v\cdot y_{\pi(v)}$ &  $w^i_v\cdot y_{\pi(v)}$ \\
\hline $\pi(v)\in E(T_2)\cup\{e\}$ & $w^i_v\cdot y_{\pi(v)}$ & $w^i_v\cdot \frac{y_{\pi(v)}}{y_e}$ & $0$ &  $w^i_v\cdot y_{\pi(v)}$ \\
\hline
\end{tabular}}
\end{center}

Next consider the second summation $P_b(y,{\cal T}) := - \frac1{4\ell}\, \sum_{i=1}^m \sum_{u,v\in V({\cal T})\cap L_i}
w^i_v \cdot c^i_u\cdot \frac{y_{\pi(u)}\cdot y_{\pi(v)}}{y_{\theta(u,v)}}$. We again show equality $P_b(y,T\cup F) =
y_e\cdot P_b(y',T')+(1-y_e)\cdot P_b(y'',T'')$ for each term $\frac{y_{\pi(u)}\cdot y_{\pi(v)}}{y_{\theta(u,v)}}$
corresponding to ``$i\in[m]$ and $u,v\in L_i$''; to reduce clutter we drop the multiplier $- \frac1{4\ell}\,w^i_v \cdot
c^i_u$.
\begin{center}
{\renewcommand{\arraystretch}{1.5}
\renewcommand{\tabcolsep}{0.2cm}\begin{tabular}{|c|c|c|c|c|}
\hline  Cases & $P_b(y,T\cup F)$ & $P_b(y',T')$ & $P_b(y'',T'')$ &
{\small \begin{tabular}[l]{@{}c@{}}$y_e\cdot P_b(y',T')$\\ $+(1-y_e)\cdot P_b(y'',T'')$\end{tabular}}\\
\hline $\pi(u), \pi(v)\in E(F)$ & $1$ & $1$ & $1$ & $1$\\
\hline $\pi(u)\in E(F),\, \pi(v)\in E(T_1)$  & $y_{\pi(v)}$ & $y_{\pi(v)}$ & $y_{\pi(v)}$ & $y_{\pi(v)}$\\
\hline $\pi(u)\in E(F),\, \pi(v)\in E(T_2)\cup\{e\}$  & $y_{\pi(v)}$ & $\frac{y_{\pi(v)}}{y_e}$ & $0$ &  $y_{\pi(v)}$\\
\hline $\pi(u), \pi(v)\in E(T_1)$  & $\frac{y_{\pi(u)}\cdot y_{\pi(v)}}{y_{\theta(u,v)}}$ & $\frac{y_{\pi(u)}\cdot y_{\pi(v)}}{y_{\theta(u,v)}}$ & $\frac{y_{\pi(u)}\cdot y_{\pi(v)}}{y_{\theta(u,v)}}$ & $\frac{y_{\pi(u)}\cdot y_{\pi(v)}}{y_{\theta(u,v)}}$\\
\hline $\pi(u), \pi(v)\in E(T_2)\cup\{e\}$  & $\frac{y_{\pi(u)}\cdot y_{\pi(v)}}{y_{\theta(u,v)}}$ & $\frac{y_{\pi(u)}\cdot y_{\pi(v)}}{y_{\theta(u,v)}\cdot y_e}$ & $0$
&  $\frac{y_{\pi(u)}\cdot y_{\pi(v)}}{y_{\theta(u,v)}}$ \\
\hline $\pi(u)\in E(T_1),\, \pi(v)\in E(T_2)\cup\{e\}$  & $y_{\pi(u)}\cdot y_{\pi(v)}$ & $y_{\pi(u)}\cdot \frac{y_{\pi(v)}}{y_e}$ & $0$ & $y_{\pi(u)}\cdot y_{\pi(v)}$\\
\hline
\end{tabular}}
\end{center}
Since we have checked all cases, it follows that $P(y,T\cup F)=y_e\cdot P(y',T')+(1-y_e)\cdot P(y'',T'')$.
\end{proof}

This lemma implies that in any iteration,
$$\max\left\{\frac{P(y',T')}{D(y',T')},\, \frac{P(y'',T'')}{D(y'',T'')}  \right\} \ge \frac{y_e\cdot P(y',T')+(1-y_e)\cdot P(y'',T'')}{y_e\cdot D(y',T')+(1-y_e)\cdot
D(y'',T'')} = \frac{P(y,T\cup F)}{D(y,T\cup F)}$$ So by induction, the ratio $\frac{P(y,T\cup F)}{D(y,T\cup F)}$ is
non-decreasing over iterations. It can be seen that the denominator $D(y,T\cup F)$ can always be taken to be non-zero and therefore the final solution $F$
must contain at least one edge. Notice that at the start of Algorithm~\ref{alg:ratio-kro}, $F$ is empty and the ratio
is at least $\rho = LP(B)/(4B)$ by~\eqref{eq:kro-exp-profit2} and~\eqref{eq:kro-exp-length}. And at the end of the
algorithm, the tree $T$ is empty-- so the ratio is exactly:
\begin{equation}\label{eq:kro-ratio-det}
\left(\sum_{i=1}^m \sum_{v\in V(F)\cap L_i} w^i_v \cdot \left( 1  - \frac1{4\ell}\,\sum_{u\in V(F) \cap L_i} c^i_u\right)\right) \quad /
\quad  \left(\sum_{f\in E(F)} d_f\right) \quad \ge \quad \rho\end{equation} The denominator is exactly the length of solution tree $F$.
We will show that the numerator is at most $O(\ell)$ times the profit $\sum_{i=1}^m f_i(V(F))$ of $F$. This would imply
that the profit/length ratio of solution $F$ is at least $\Omega(\frac1\ell)\cdot \rho$ as desired.

To upper bound the numerator in~\eqref{eq:kro-ratio-det}, define
$$ R_i =\left\{ \begin{array}{ll} L_i\cap V(F) & \mbox{
if } c^i(L_i\cap V(F))\le 4\ell\\
\emptyset & \mbox{ otherwise}.
\end{array}\right.\qquad \forall i\in[m]$$

Note that if $R_i=\emptyset$ then $c^i(L_i\cap V(F))> 4\ell$, i.e. the contribution of $L_i\cap V(F)$ in the numerator
of~\eqref{eq:kro-ratio-det} is negative. On the other hand, if $R_i\ne \emptyset$ then the contribution of $L_i\cap V(F)$
is at most $w^i(L_i\cap V(F))=w^i(R_i)$. So we obtain that the numerator of~\eqref{eq:kro-ratio-det} is at most
$\sum_{i=1}^m w^i(R_i)$. Since each $R_i$ has knapsack-size at most $4\ell$, a greedy partitioning as before implies
there is a subset $R'_i\sse R_i$ with size $c^i(R'_i)\le 1$ and $w^i(R'_i)\ge w^i(R_i)/(8\ell)$; i.e. $f_i(V(F))\ge
w^i(R_i)/(8\ell)$. Rearranging, the numerator of~\eqref{eq:kro-ratio-det} is at most $8\ell\,\sum_{i=1}^m f_i(V(F))$.
Combined with the inequality in~\eqref{eq:kro-ratio-det},
$$\mbox{Ratio of solution }F \,\, = \,\, \frac{f(V(F))}{d(V(F))} \,\,\ge \,\, \frac1{32\,\ell}\cdot \frac{LP(B)}{ B}.$$

Thus we have proved:
\begin{theorem}
There is a deterministic $O(\ell)$-approximation algorithm for the ratio knapsack orienteering problem on depth $\ell$
trees. On general metrics there is an $O(\log^2 n)$-approximation algorithm.
\end{theorem}
The additional log-factor on general metrics is due to tree embedding~\cite{FRT04} which is randomized. This step can
also be made deterministic using the algorithm in~\cite{CCGG98}.

\section{UGC Hardness of Approximation}\label{sec:hardness}
In this section we
prove a $\omega(1)$ UGC-hardness of approximation for \svrp even for a very simple star-like metric
with a setting of $\lambda$ that renders the recourse tour trivial. Our hardness result is based on the Unique
Games Conjecture (UGC) of Khot \cite{Khot02} , a restatement of which is given below.
\begin{conjecture} (Unique Games Conjecture \cite{Khot02})
For any $\varepsilon > 0$, there is a positive integer $p$ such
that:  given a system of $2$-variable linear equations over $\mathbb{Z}_p$,
each of the form $x_i - x_j = a_{ij} \mod p$, it is NP-hard to distinguish
between the following two cases : (i) \textnormal{YES CASE:} There is
an assignment to the variables that satisfies $1- \varepsilon$
fraction of equations, (ii) \textnormal{NO CASE:} Any assigment
satisfies at most $\varepsilon$ fraction of equations.
\end{conjecture}
Based on UGC, Bansal and Khot \cite{BK10} proved the following
hardness of approximation result for minimum vertex cover on almost
$k$-partite $k$-uniform hypergraphs, which shall be the starting point of our
reduction.
\begin{theorem}
\label{thm-BK10}\cite{BK10} Assuming the Unique Games Conjecture, for any
$\varepsilon
> 0$ and positive integer $k \geq 2$, given a $k$-uniform hypergraph
$G$ with vertex set $U$ and hyperedge set $E$, it is NP-hard to distinguish
between the following two cases:

\medskip
\noindent \textnormal{YES CASE:} There is a partition of $U$ into $k+1$ disjoint subsets $X, U_1, \dots, U_k$ such that
$|X| \leq \varepsilon|U|$ and the hypergraph induced by $U\setminus X$
(consisting of vertex set $U\setminus X$ and hyperedge set $\{e\cap (U\setminus X)\ \mid\ e \in E, |e\cap (U\setminus X)|>0\}$)
is $k$-partite with $U_1, \dots, U_k$ as the
$k$-partition. That is, any hyperedge $e \in E$ has at most one vertex from any $U_i$. This implies that $X\cup U_i$ is a
vertex cover in $G$ for each $i = 1, \dots, k$, and that the minimum vertex cover in $G$ has size at most $(1/k +
\varepsilon)|U|$.

\medskip
\noindent
\textnormal{NO CASE:} The size of the maximum independent set in $G$
is at most $\varepsilon|U|$, and therefore the size of the minimum vertex
cover in $G$ is at least $(1 -\varepsilon)|U|$.
\end{theorem}
In the rest of this section we shall give a hardness reduction from the problem of distinguishing between $k$-uniform
hypergrahs which are almost $k$-partite (as in the YES case of Theorem \ref{thm-BK10}) from those that have a very
small maximum independent set (as in the NO case of Theorem \ref{thm-BK10}).
\subsection{Hardness Reduction}
Fix any positive integer $k \geq 2$.
Let us suppose we are given a $k$-uniform hypergraph $G$ on vertex set
$U$ and with hyperedge set $E$ as a
hard instance from Theorem \ref{thm-BK10}, where we shall fix the
parameter $\varepsilon$ in Theorem \ref{thm-BK10} later. We transform
$G(U, E)$ into an instance of \svrp as follows. For clarity, in this
section the nomenclature of ``vertices'' shall be in context of the
hypergraph, while ``points'' shall be used for corresponding elements
in the metric.

\medskip
\noindent
{\em Metric $(V, d)$.} The set of points $V$ in the metric is $U\cup\{r\}$, where $r$ is the {\it root}. The distances
$d$ are defined as follows. Let $d(r, u) = L$, where $L = (|U|/2k + 1/2)$, for all $u \in U$. Further, for each pair
$u, u' \in U$, $u \neq u'$, let $d(u, u') = 1$. It is easy to see that $d$ is a metric. This simple metric can be
realized by the shortest paths in a star-like tree of distances as illustrated in Figure \ref{fig-metric}.

\medskip
\noindent
{\em Capacity and Demands.} The capacity $Q=1$ and demands
will be $\{0,1\}$.

\medskip
\noindent
{\em Demand Distribution \D.} There are polynomially many scenarios $m=|E|$, each having uniform probability. Every
hyperedge $e \in E$ is a scenario having demand of one at all points in $e$, and zero demand elsewhere.

\medskip
\noindent
{\em Parameter $\lambda$.} We set $\lambda = 2m|U|(k+1)$.

\medskip
Before we proceed to the analysis of this reduction, we note that the cost of the minimum cost $r$-tour covering points
$S \subseteq V\setminus\{r\}$, is simply $|U|/k + |S|$. Also, the optimal value is at most $\lambda/m$. Consider the
fixed tour consisting of $k$ identical $r$-tours each covering $U$: since each scenario has at most $k$ demands, this
solution never uses a recourse tour, and has cost $k\cdot (|U|/k+|U|)<\lambda/m$. So we may assume that {\em the
optimal solution has no recourse tour}: if the recourse tour is non-empty in any scenario then its cost is at least
$\lambda/m$.

\begin{figure}
\begin{center}\includegraphics[scale=0.5]{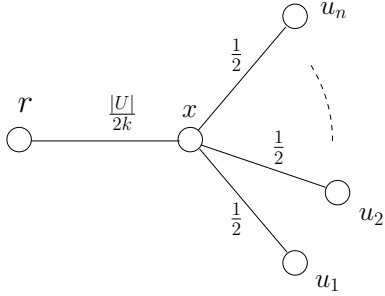}
\caption{Tree of distances realizing metric $d$, with intermediate point $x$ and $V = \{r, u_1, \dots, u_n\}$.}
\label{fig-metric}
\end{center}
\end{figure}

\subsection{ Analysis}
We now give the analysis.

\medskip
\noindent
{\em YES Case.} Suppose that $G(U, E)$ is a YES instance of Theorem \ref{thm-BK10} with $X, U_1, \dots, U_k$ as the
partition of $U$ with the properties as stated in the theorem. Consider the $r$-tours $\tau_1, \dots,$ $\tau_k$, where
$\tau_i$ is an $r$-tour that covers points $X\cup U_i$ (in addition to $r$). Since every scenario in our instance of
\svrp corresponds to a hyperedge in $G$, using the property in the YES case that each hyperedge has at most one vertex
from each $U_i$, we see that the $r$-tours $\tau_1,\dots, \tau_k$ satisfy all the scenarios. As noted earlier the cost
of each $r$-tour that covers $S\subseteq V\setminus \{r\}$ is $|U|/k + |S|$. Therefore the total cost of the $k$
$r$-tours $\tau_1, \dots, \tau_k$ is,
$$k\cdot(|U|/k) + \sum_{i=1}^k|X\cup U_i|\ \leq\ |U|
+ (1 + k\varepsilon)|U|\ =\ (2 + k\varepsilon)|U|,$$
by the properties of the partition $X, U_1,\dots, U_k$ of $U$.

\medskip
\noindent
{\em NO Case.} Suppose that $G(U, E)$ is a NO instance of Theorem \ref{thm-BK10}, so that the maximum independent set
in $G$ is of size at most $\varepsilon |U|$. In this case we shall prove that the total cost of any set of $r$-tours
that satisfy all scenarios is at least $k(1 - f_k(\varepsilon))|U|$, where $f_k(\varepsilon) \to 0$ as $\varepsilon \to
0$ for any fixed positive integer $k \geq 2$. We may assume that the number of $r$-tours in the optimal solution is at
most $k^2$, otherwise the total cost will be at least $k^2(|U|/k) = k|U|$ and we shall be done. Therefore, let
$\gamma_1, \dots, \gamma_T$ be the $r$-tours in an optimal fixed tour, where  $T \leq k^2$. We shall estimate the
number of points in $U$ which occur in at most $k-1$ of these $r$-tours. For any subset $I \subseteq [T]$, let $U(I)
\subseteq U$ be the points which {\em do not} occur in $\{\gamma_i: i \in [T]\setminus I\}$. We have the following
simple lemma.
\begin{lemma} For any $I \subseteq [T]$ with $|I| = k -1$, $U(I)$
is an independent set in $G$.
\end{lemma}
\begin{proof} For a contradiction, suppose that $e$ is a hyperedge
induced by $U(I)$. Since $|e| = k$, the scenario corresponding to $e$ will not be satisfied by our solution as the $k$
vertices of $e$ appear (as points) in at most $k-1$ of the $r$-tours, namely those given by $I \subseteq [T]$. Recall
that each $r$-tour can serve only one demand.
\end{proof} The total number of points in $U$ that appear in at most $k-1$ of the $r$-tours is upper bounded by,
$$\sum_{I\subseteq [T], |I| = k-1} |U(I)|.$$
There are ${T\choose k-1} \leq 2^T \leq 2^{k^2}$ choices for the subsets $I$ in the above expression. Using the fact
that any independent set in $G$ has size at most $\varepsilon |U|$, the fraction of points in $U$ that occur in at most
$k-1$ of the $r$-tours is at most $\varepsilon2^{k^2} =: f_k(\varepsilon)$. Each of the remaining
$(1-f_k(\varepsilon))|U|$ points appears in at least $k$ of the $r$-tours; so the total cost of the fixed tour is $k(1 -
f_k(\varepsilon))|U|$.

\medskip
\noindent {\bf Hardness Factor.} In the YES case there is a solution of cost at most $(2 + k\varepsilon)|U|$, whereas
in the NO case any solution has cost at least $k(1 - f_k(\varepsilon))|U|$. For any positive integer $k \geq 2$ and
arbitrarily small $\delta > 0$, choosing $\varepsilon > 0$ to be small enough in Theorem \ref{thm-BK10}, we obtain a
hardness factor of $k/2 - \delta$.

\section{Independent Demand Distributions}\label{app:indep} In this section we give an
$O(\log (n\lambda)/\llog (n\lambda))$-approximation for \svrp under independent distributions $\D$. That is, the demand
$\oq_v$ at each vertex $v$ is independent of all other vertices $V\setminus\{v\}$. The main idea is to show the
existence of a near-optimal solution that partitions vertices into two disjoint sets $D_1$ and $D_2$ such that:
vertices $D_1$ are served by the fixed tour w.h.p., and  vertices $D_2$ are served in the recourse tour.  This step
(Lemma~\ref{lem:indep-struc}) uses independence. Then we show how an LP-based approach (combined with sampling) yields
a constant approximation to the problem of choosing the best partition $(D_1,D_2)$. For any $v\in V$ let
$\mu_v:=\Ex[\oq_v]$ the expected demand at $v$; note $\max_{v\in V}\mu_v\le Q$ the vehicle capacity. We assume by
scaling that the optimal value $\OPT\ge 1$.

\begin{lemma}\label{lem:indep-struc}
Given any instance of \svrp with independent demands, there exists partition $D_1\cup D_2=V$ and \svrp solution with
fixed tour $\tau$ such that:
\begin{OneLiners}
\item The total expected demand in each $r$-tour of $\tau$ is at most $Q$.
\item The length of $\tau$ is $O(\log (n\lambda)/\llog (n\lambda))\cdot \OPT$.
\item $\tau$ does not visit any $D_2$-vertex; i.e. each $u\in D_2$ is served in recourse tour.
\item Each $v\in D_1$ is served by $\tau$ with probability at least $1-1/(n\lambda)^4$.
\item The recourse cost is at most $\OPT+1$.
\end{OneLiners}
\end{lemma}
\begin{proof}
Consider an optimal fixed tour $\tau^*$. Let $D_1\sse V$ denote the vertices visited at least once in $\tau^*$; note
that each vertex might be visited multiple times. Clearly the minimum spanning tree on vertices $D_1\cup\{r\}$,
$\mst(D_1)\le d(\tau^*)\le \OPT$. Using the ``flow lower bound'' in VRP~\cite{HK85} it is also clear that:
$$\OPT\,\, \ge \,\, \frac1Q \sum_{v\in V} d(r,v)\cdot \mu_v \,\, \ge \,\, \frac1Q \sum_{v\in D_1} d(r,v)\cdot \mu_v \,\, = \,\, \flow(D_1)$$
Recall that each $\mu_v\le Q$. Thus if we consider a deterministic VRP instance with demands $\{\mu_v : v\in D_1\}$ and
capacity $Q$, then it has a solution $\tau'$ of length at most $O(1)\cdot (\mst(D_1) + \flow(D_1))$
by~\cite{HK85,AG87}. From the above, we have $d(\tau')\le O(1)\cdot \OPT$. Let $\tau'_1,\ldots,\tau'_t$ denote the
$r$-tours in $\tau'$, each having total $\mu$-value at most $Q$. We define the fixed tour $\tau$ to consist of $\beta:=
c\cdot \frac{\log (n\lambda)}{\llog (n\lambda)}$ {\em copies} of $\tau'$, where $c$ is a large enough constant. The
first three properties of $\tau$ are immediate. For any $r$-tour $\tau'_i$, if the total instantiated demand
$\sum_{v\in\tau'_i}\oq_v\le \beta\cdot Q$ then all these demands can be served by $\tau$ since it contains $\beta$
copies of $\tau'_i$. Thus the probability that some $v\in D_1$ (say with $v\in \tau'_i$) is not served by $\tau$ is:
$$\Pr[v\mbox{ not covered by }\tau] \quad \le \quad \Pr\left[ \sum_{v\in\tau'_i}\oq_v > \beta\cdot Q\right] \quad \le \quad \frac1{(n\lambda)^4},$$
by a Chernoff bound~\cite{MR95} using the fact that
$\beta=\Theta\left(\frac{\log (n\lambda)}{\llog (n\lambda)}\right)$. This proves the fourth property. For the final
property, note that vertices $D_2$ that are never served by the optimal fixed tour $\tau^*$. So the expected VRP value
on $D_2$ (scaled by $\lambda$) is at most $\OPT$. This is the recourse cost that our solution corresponding to $\tau$
pays for $D_2$. In addition, some $D_1$-vertices may be uncovered in $\tau$ and the recourse cost due to these is at
most:
$$\sum_{v\in D_1} \lambda\cdot 2d(r,v)\cdot \Pr[v\mbox{ not covered by }\tau] \quad \le \quad \frac{2n\lambda
D}{(n\lambda)^4}\quad \le \quad 1,$$
where we used the fact that diameter $D=O(n^3)$ from
Subsection~\ref{app:redn-explicit}. So the total expected recourse cost is at most $\OPT+1$ as claimed.
\end{proof}

We now find an approximately optimal solution to independent \svrp that has the above structure. We write an IP
formulation to capture the partition $(D_1,D_2)$. For $v\in V$ let $x_v\in\{0,1\}$ denote the indicator that $v\in
D_1$. By Lemma~\ref{lem:indep-struc} the fixed tour $\tau$ corresponds to a deterministic VRP solution with demands
$\{\mu_v\cdot x_v : v\in V\}$. Using MST and flow bounds (as in Section~\ref{sec:black-box}) we can express this
(losing a constant factor) via linear constraints in $x$.

We also need to write the expected VRP value (scaled by $\lambda$) due to demands $D_2$. This involves the expected VRP
value (equivalently $\Ex \mst+\Ex \flow$) of the random instance where each $v\in V$ has an independent demand of
$(1-x_v)\cdot \oq_v$; recall that $\oq_v$ denotes $v$'s demand in the original \svrp instance. The expectation
$\Ex\flow$ is just $\frac\lambda{Q} \sum_{v\in V} d(r,v)\cdot \mu_v\cdot (1-x_v)$. Unfortunately it is not clear if one
can write linear constraints (in $x$) for the expectation of $\mst$: this involves the expected MST value when each
$v\in V$ is present independently with probability $(1-x_v)\cdot \Pr[\oq_v>0]$. Instead we show that sampling  can be
used to estimate $\Ex\mst$ within small error, and that the sample expectation can be expressed via linear constraints
in $x$.

For any $x\in \{0,1\}^V$ define $T(x) := \Ex[\mst(S_x)]$ where $S_x$ contains each vertex $v\in V$ independently w.p.
$(1-x_v)\cdot p_v$ where $p_v:=\Pr[\oq_v>0]$. We now use the result of~\cite{CCP05} as in Theorem~\ref{thm:sampling}.
We make $m=poly(n,\lambda)$ independent samples $S^1,\ldots,S^m\sse V$ according to $\{p_v\}_{v\in V}$ and set
$$\hat{T}(x):= \frac1m\sum_{i=1}^m \mst\left(\{v\in S^i : x_v=0\}\right).$$
Then we have $|T(x)-\hat{T}(x)|\le 1$ for all $x\in \{0,1\}^V$ with probability $1-o(1)$.

We now write the following integer program for finding partition $(D_1,D_2)$.
\begin{alignat*}{10} &\min \,\,
\sum_{e\in E} d_e\cdot z_e \,  +  \, \sum_{v\in V} \frac{d(r,v) \,\mu_v}{Q} \cdot x_v  \, +  \, \sum_{v\in V}
\frac{\lambda\, d(r,v)\, \mu_v} {Q}\cdot (1-x_v)\,  +  \,
\frac{\lambda}{m}  \sum_{i=1}^m \sum_{e\in E} z_{e}^{i} \cdot d_e \notag \\
&\mbox{s.t.}\quad \sum_{e\in \delta(R)} z_e \ge x_v \quad \quad \quad \forall R\sse V\setminus \{r\},\,\, \forall v\in
R,\\
&\quad \sum_{e\in \delta(R)} z_{e}^{i} \ge 1-x_v \quad \quad \quad \forall i\in[m],\,\, \forall R\sse V\setminus \{r\},
\,\, \forall v\in R\cap S^i, \\
&\quad x_v \in \{0,1\} \qquad \qquad  \qquad  \forall v\in V,\notag \\
&\quad z_e,\, z_{e}^{i}\ge 0 \qquad \qquad \qquad  \quad \forall e\in E, \forall i \in [m].\notag
\end{alignat*}
The last term in the objective captures $\hat{T}(x)$. Based on the preceding discussion, w.h.p. this integer program
expresses the objective of {\em all} partitions $(D_1,D_2)$ up to a constant factor. Relaxing the integrality on $x$ we
obtain an LP relaxation that can be solved in polynomial time. The rounding algorithm simply chooses $D_1=\{v\in V :
x_v>\frac12\}$ and $D_2=V\setminus D_1$. We output the fixed tour $\tau$ to be $O(\log (n\lambda)/\llog (n\lambda))$
copies of an approximate VRP on demands $\{\mu_v: v\in D_1\}$. The recourse step involves greedily satisfying the
instantiated demands on $\tau$, and then computing an approximate VRP solution on the residual demands. Using
Lemma~\ref{lem:indep-struc} it is easy to show that this achieves an $O(\log (n\lambda)/\llog
(n\lambda))$-approximation, i.e. Theorem~\ref{thm:indep}.

\bibliographystyle{plain}
\bibliography{2-stg-vrp}

\begin{thebibliography}{10}

\bibitem{AE07}
A.~Ak and A.L. Erera.
\newblock A paired-vehicle recourse strategy for the vehicle-routing problem
  with stochastic demands.
\newblock {\em Transportation Science}, 41(2):222--237, 2007.

\bibitem{AG87}
K.~Altinkemer and B.~Gavish.
\newblock Heuristics for unequal weight delivery problems with a fixed error
  guarantee.
\newblock {\em Operations Research Letters}, 6:149--158, 1987.

\bibitem{BK10}
N.~Bansal and S.~Khot.
\newblock Inapproximability of hypergraph vertex cover and applications to
  scheduling problems.
\newblock In {\em ICALP (1)}, pages 250--261, 2010.

\bibitem{B92}
D.J. Bertsimas.
\newblock { A vehicle routing problem with stochastic demand}.
\newblock {\em Operations Research}, 40(3):574--585, 1992.

\bibitem{BJO90}
D.J. Bertsimas, P.~Jaillet, and A.R. Odoni.
\newblock A priori optimization.
\newblock {\em Operations Research}, 38(6):1019--1033, 1990.

\bibitem{BL97}
J.R. Birge and F.~Louveaux.
\newblock {\em Introduction to Stochastic Programming}.
\newblock Springer-Verlag, New York, 1997.

\bibitem{CZ05}
G.~Calinescu and A.~Zelikovsky.
\newblock The polymatroid steiner problems.
\newblock {\em Journal of Combinatorial Optimization}, 9(3):281--294, 2005.

\bibitem{CT08}
A.M. Campbell and B.W. Thomas.
\newblock Probabilistic traveling salesman problem with deadlines.
\newblock {\em Transportation Science}, 42(1):1--21, 2008.

\bibitem{CCGG98}
M.~Charikar, C.~Chekuri, A.~Goel, and S.~Guha.
\newblock Rounding via trees: Deterministic approximation algorithms for group
  steiner trees and {\it k}-median.
\newblock In {\em STOC}, pages 114--123, 1998.

\bibitem{CCP05}
M.~Charikar, C.~Chekuri, and M.~P{\'a}l.
\newblock Sampling bounds for stochastic optimization.
\newblock In {\em APPROX-RANDOM}, pages 257--269, 2005.

\bibitem{CEK06}
C.~Chekuri, G.~Even, and G.~Kortsarz.
\newblock A greedy approximation algorithm for the group steiner problem.
\newblock {\em Discrete Applied Mathematics}, 154(1):15--34, 2006.

\bibitem{CP05}
C.~Chekuri and M.~P{\'a}l.
\newblock A recursive greedy algorithm for walks in directed graphs.
\newblock In {\em FOCS}, pages 245--253, 2005.

\bibitem{DGV08}
B.C. Dean, M.X. Goemans, and J.~Vondr{\'a}k.
\newblock Approximating the stochastic knapsack problem: The benefit of
  adaptivity.
\newblock {\em Math. Oper. Res.}, 33(4):945--964, 2008.

\bibitem{D05}
M.~Dror.
\newblock {Vehicle Routing with Stochastic Demands: Models \& Computational
  Methods}.
\newblock In {\em Modeling Uncertainty: International Series In Operations
  Research \& Management Science}, volume 46(8), pages 625--649. Springer,
  2005.

\bibitem{ESU09}
A.L. Erera, M.W.P. Savelsbergh, and E.~Uyar.
\newblock Fixed routes with backup vehicles for stochastic vehicle routing
  problems with time constraints.
\newblock {\em Networks}, 54(4):270--283, 2009.

\bibitem{FRT04}
J.~Fakcharoenphol, S.~Rao, and K.~Talwar.
\newblock A tight bound on approximating arbitrary metrics by tree metrics.
\newblock {\em J. Comput. Syst. Sci.}, 69(3):485--497, 2004.

\bibitem{GKR00}
N.~Garg, G.~Konjevod, and R.~Ravi.
\newblock { A Polylogarithmic Approximation Algorithm for the Group Steiner
  Tree Problem}.
\newblock {\em Journal of Algorithms}, 37(1):66--84, 2000.

\bibitem{GKNR12}
A.~Gupta, R.~Krishnaswamy, V.~Nagarajan, and R.~Ravi.
\newblock {Approximation Algorithms for Stochastic Orienteering}.
\newblock In {\em SODA}, pages 245--253, 2012.

\bibitem{GNR12}
A.~Gupta, V.~Nagarajan, and R.~Ravi.
\newblock Approximation algorithms for vrp with stochastic demands.
\newblock {\em Operations Research}, To Appear.

\bibitem{GPRS11}
A.~Gupta, M.~P{\'a}l, R.~Ravi, and A.~Sinha.
\newblock Sampling and cost-sharing: Approximation algorithms for stochastic
  optimization problems.
\newblock {\em SIAM J. Comput.}, 40(5):1361--1401, 2011.

\bibitem{GS06}
A.~Gupta and A.~Srinivasan.
\newblock An improved approximation ratio for the covering steiner problem.
\newblock {\em Theory of Computing}, 2(1):53--64, 2006.

\bibitem{HK85}
M.~Haimovich and A.~H. G.~Rinnooy Kan.
\newblock Bounds and heuristics for capacitated routing problems.
\newblock {\em Mathematics of Operations Research}, 10(4):527--542, 1985.

\bibitem{IKMM04}
N.~Immorlica, D.R. Karger, M.~Minkoff, and V.S. Mirrokni.
\newblock On the costs and benefits of procrastination: approximation
  algorithms for stochastic combinatorial optimization problems.
\newblock In {\em SODA}, pages 691--700, 2004.

\bibitem{KKU08}
I.~Katriel, C.K. Mathieu, and E.~Upfal.
\newblock Commitment under uncertainty: Two-stage stochastic matching problems.
\newblock {\em Theor. Comput. Sci.}, 408(2-3):213--223, 2008.

\bibitem{Khot02}
S.~Khot.
\newblock On the power of unique 2-prover 1-round games.
\newblock In {\em STOC}, pages 767--775, 2002.

\bibitem{KRS02}
G.~Konjevod, R.~Ravi, and A.~Srinivasan.
\newblock Approximation algorithms for the covering steiner problem.
\newblock {\em Random Struct. Algorithms}, 20(3):465--482, 2002.

\bibitem{LST90}
J.K. Lenstra, D.B. Shmoys, and E.~Tardos.
\newblock Approximation algorithms for scheduling unrelated parallel machines.
\newblock {\em Mathematical Programming}, 46:259--271, 1990.

\bibitem{MR95}
R.~Motwani and P.~Raghavan.
\newblock {\em Randomized Algorithms}.
\newblock Cambridge University Press, 1995.

\bibitem{RS06}
R.~Ravi and A.~Sinha.
\newblock Hedging uncertainty: Approximation algorithms for stochastic
  optimization problems.
\newblock {\em Math. Program.}, 108(1):97--114, 2006.

\bibitem{SG95}
M.W.P. Savelsbergh and M.~Goetschalkx.
\newblock A comparison of the efficiency of fixed versus variable vehicle
  routes.
\newblock {\em J. Business Logistics}, 16:163--187, 1995.

\bibitem{Schrijver}
A.~Schrijver.
\newblock {\em Combinatorial optimization: polyhedra and efficiency}.
\newblock Springer-Verlag, Berlin, 2003.

\bibitem{ST08}
D.~Shmoys and K.~Talwar.
\newblock {A Constant Approximation Algorithm for the a priori Traveling
  Salesman Problem}.
\newblock In {\em IPCO}, pages 331--343, 2008.

\bibitem{SS07}
D.B. Shmoys and M.~Sozio.
\newblock Approximation algorithms for 2-stage stochastic scheduling problems.
\newblock In {\em IPCO}, pages 145--157, 2007.

\bibitem{SS06}
D.B. Shmoys and C.~Swamy.
\newblock An approximation scheme for stochastic linear programming and its
  application to stochastic integer programs.
\newblock {\em J. ACM}, 53(6):978--1012, 2006.

\bibitem{Srinivasan01}
A.~Srinivasan.
\newblock New approaches to covering and packing problems.
\newblock In {\em SODA}, pages 567--576, 2001.

\bibitem{SG83}
W.~Stewart and B.~Golden.
\newblock Stochastic vehicle routing: A comprehensive approach.
\newblock {\em Eur. Jour. Oper. Res.}, 14:371--385, 1983.

\bibitem{TV01}
P.~Toth and D.~Vigo.
\newblock {\em The vehicle routing problem}.
\newblock 2001.

\end{thebibliography}

\end{document}